\documentclass{amsart}
\usepackage{amssymb}
\usepackage{lmodern}
\usepackage[T1]{fontenc}
\usepackage[utf8]{inputenc}
\usepackage{microtype}
\usepackage[usenames, dvipsnames]{xcolor}
\usepackage{tikz}
\usetikzlibrary{decorations.markings,arrows,decorations.pathreplacing,arrows,decorations.pathmorphing,decorations.pathreplacing,decorations.markings,shapes.geometric,through,fit,shapes.symbols}
\tikzset{ribbon/.style={draw=white,double=purple,very thick,double distance=2pt}}
\tikzset{onearrowlate/.style={postaction={decorate}, decoration={markings,mark=at position .7 with {\arrow[draw,line width=1pt]{>}}}}}
\tikzset{onearrow/.style={postaction={decorate}, decoration={markings,mark=at position .55 with {\arrow[draw,line width=1pt]{>}}}}}
\tikzset{condensation/.style={-,
          decoration={markings, 
                    mark=at position -1ex with {\arrow{left hook}},
                    mark=at position 1 with {\arrow[scale=1.5]{>}},
                   },
        preaction={decorate}}}
\tikzset{condensationlarge/.style={-,
          decoration={markings, 
                    mark=at position -1.5ex with {\arrow[scale=1.25]{left hook}},
                    mark=at position 1 with {\arrow[scale=1.75]{>}},
                   },
        preaction={decorate}}}
\tikzset{doublecondensation/.style={double,
        double distance=1.5pt,
        shorten <=6pt, shorten >=6pt, 
        decoration={markings, 
                    mark=at position -11pt with {\arrow[scale=1.25]{left hook}},
                    mark=at position -4.25pt with {\arrow[scale=1.5]{>}},
                   },
        preaction={decorate}}}
\tikzset{doublecondensationlarge/.style={double,
        double distance=1.5pt,
        shorten <=6pt, shorten >=6pt, 
        decoration={markings, 
                    mark=at position -11pt with {\arrow[scale=1.5]{left hook}},
                    mark=at position -4.25pt with {\arrow[scale=2]{>}},
                   },
        preaction={decorate}}}

\usepackage{tikz-3dplot}
\usepackage{tikz-cd}
\usepackage[margin=1.3in]{geometry}
\usepackage[backref=page, colorlinks=true]{hyperref}
\hypersetup{
 linkcolor={red!50!black},
 citecolor={green!50!black},
 urlcolor={blue!80!black}
}
\usepackage[capitalize, noabbrev]{cleveref} 
\usepackage{tikz-3dplot}

\usepackage{thmtools}

\numberwithin{equation}{section}
\usepackage{amsthm}
\newtheorem{thm}[equation]{Theorem}
\newtheorem{lem}[equation]{Lemma}
\newtheorem{prop}[equation]{Proposition}
\newtheorem{assume}[equation]{Assumption}
\newtheorem{cor}[equation]{Corollary}

\newtheorem*{thm*}{Theorem}

\newcounter{mainthm}
\newtheorem{maintheorem}[mainthm]{Theorem}

\newcounter{claim}
\newtheorem{clm}[claim]{Claim}

\newtheorem*{notation*}{Notation}
\newtheorem{notation}{Notation}
\theoremstyle{definition}
\newtheorem{defn}[equation]{Definition}
\newtheorem{example}[equation]{Example}

\theoremstyle{remark}
\newtheorem{rem}[equation]{Remark}
\newtheorem{question}[equation]{Question}

\crefname{thm}{Theorem}{Theorems}
\crefname{prop}{Proposition}{Propositions}
\crefname{lem}{Lemma}{Lemmas}
\crefname{cor}{Corollary}{Corollaries}
\crefname{defn}{Definition}{Definitions}
\usepackage{mathtools}

\usepackage{tikz-cd}
\usepackage{subcaption}
\usepackage{setspace}
\usetikzlibrary{calc}

\usepackage[mathscr]{eucal}
\makeatletter
\def\instring#1#2{TT\fi\begingroup
  \edef\x{\endgroup\noexpand\in@{#1}{#2}}\x\ifin@}
\def\isuppercase#1{%
  \instring{#1}{ABCDEFGHIJKLMNOPQRSTUVWXYZ}%
}%
\newcommand{\C@lIfUpper}[1]{
 \if\isuppercase{#1}\mathscr{#1}%
 \else #1%
 \fi
}
\newcommand{\cat}[1]{\mathit{\@tfor\next:=#1\do{\C@lIfUpper{\next}}}}
\makeatother

\newcommand{\meridian}[2]{%
  \draw[line width=0.65pt, #2, rotate={#1}]
       (3.10,0) arc[start angle=0, end angle=180,
                    x radius=1.15, y radius=0.10];
  \draw[line width=0.65pt, #2, dashed, dash pattern=on 2pt off 2pt, rotate={#1}]
       (3.10,0) arc[start angle=0, end angle=-180,
                    x radius=1.15, y radius=0.10];
}

\usepackage{xparse}
\newcommand{\Z}{\mathbb Z}

\newcommand{\rB}{\mathrm B}
\newcommand{\C}{\mathbb C}

\newcommand{\pt}{\mathrm{pt}}

\newcommand{\id}{\mathrm{id}}

\newcommand{\condense}{\mathrel{\,\hspace{.75ex}\joinrel\rhook\joinrel\hspace{-.75ex}\joinrel\rightarrow}}

\newcommand{\cA}{\mathcal A}
\newcommand{\cB}{\mathcal B}
\newcommand{\cC}{\mathcal C}
\newcommand{\cD}{\mathcal D}
\newcommand{\cF}{\mathcal F}
\newcommand{\cG}{\mathcal G}

\newcommand{\cI}{\mathcal I}

\newcommand{\cO}{\mathcal O}
\newcommand{\cR}{\mathcal R}
\newcommand{\cS}{\mathcal S}
\newcommand{\cZ}{\mathcal Z}

\DeclareMathOperator{\Ad}{Ad}

\newcommand{\cE}{\mathcal E}

\newcommand{\colim}{\mathrm{colim}}
\newcommand{\Vect}{\mathbf{Vect}}

\newcommand{\End}{\mathrm{End}}
\DeclareMathOperator{\tVect}{\mathbf{2Vect}}

\newcommand{\op}{\mathrm{op}}
\newcommand{\cQ}{\mathcal{Q}}
\newcommand{\bfQ}{\mathbf{Q}}

\newcommand{\cT}{\mathcal{T}}
\newcommand{\sWitt}{s\mathcal{W}itt}
\newcommand{\Witt}{\mathcal{W}itt}

\newcommand{\QS}{\mathbf{QS}}
\newcommand{\Net}{\mathbf{Net}}
\newcommand{\Disk}{\mathbf{Disk}}
\newcommand{\Alg}{\mathbf{Alg}}
\newcommand{\fd}{\mathrm{fd}}
\newcommand{\npVect}{\mathbf{(n+1)Vect}}
\usetikzlibrary{arrows.meta,decorations.pathreplacing,calc,positioning}

\definecolor{leftblue}{RGB}{82,133,194}
\definecolor{midgreen}{RGB}{119,179,99}
\definecolor{rightred}{RGB}{219,109,109}
\definecolor{fatgray}{RGB}{235,235,235}
\definecolor{cA}{RGB}{88,141,195}
\definecolor{cB}{RGB}{124,178,111}
\definecolor{cC}{RGB}{234,164,95}
\definecolor{cD}{RGB}{209,112,132}
\definecolor{frame}{RGB}{35,35,35}
\definecolor{c0}{RGB}{92,143,197}
\definecolor{c1}{RGB}{128,182,110}
\definecolor{c2}{RGB}{236,170,88}
\definecolor{c3}{RGB}{215,116,135}
\definecolor{frame}{RGB}{35,35,35}
\definecolor{lightgreen}{RGB}{200,240,200}
\pgfdeclarehorizontalshading{surfshade}{8cm}{%
  rgb(0cm)=(0.40,0.03,0.05);
  rgb(1.1cm)=(0.70,0.10,0.10);
  rgb(2.3cm)=(0.99,0.80,0.78);
  rgb(3.3cm)=(0.95,0.55,0.52);
  rgb(4.5cm)=(0.82,0.18,0.18);
  rgb(6cm)=(0.42,0.04,0.05)}

\newcommand\MAILTO[1]{\href{mailto:#1}{\nolinkurl{#1}}}

\DeclareDocumentCommand{\shortexact}{s O{} O{} mmmm}{
\IfBooleanTF{#1}{ 
\begin{tikzcd}[ampersand replacement=\&]
	{1} \& {#4} \& {#5} \& {#6} \& {1#7}
	\arrow[from=1-1, to=1-2]
	\arrow["#2", from=1-2, to=1-3]
	\arrow["#3", from=1-3, to=1-4]
	\arrow[from=1-4, to=1-5]
\end{tikzcd}
}{ 
\begin{tikzcd}[ampersand replacement=\&]
	{0} \& {#4} \& {#5} \& {#6} \& {0#7}
	\arrow[from=1-1, to=1-2]
	\arrow["#2", from=1-2, to=1-3]
	\arrow["#3", from=1-3, to=1-4]
	\arrow[from=1-4, to=1-5]
\end{tikzcd}
}}

\usepackage{xspace}

\newcommand{\Aut}{\mathrm{Aut}}

\title{Categorifying Quantum Spin Systems}
\date{\today}

\author{Matthew Yu}
\address{School of Natural Sciences, Institute for Advanced Study,  1 Einstein Drive, Princeton, NJ 08540, USA}
\email{\href{mailto:yumatthew70@gmail.com}{yumatthew70@gmail.com}}
\thanks{It is a pleasure to thank 
Ryohei Kobayashi,
Claudia Scheimbauer,
Devon Stockall,
Siddarth Vadnerkar,
Jackson Van Dyke,
and Bowen Yang
for helpful discussions.
MY is supported by the EPSRC Open Fellowship EP/X01276X/1, and the Fund for Natural Science and the Adler Family Fund.}

\begin{document}
\maketitle

\begin{abstract}
We develop a homotopical framework for understanding the space of quantum spin systems based on topological $n$-nets, which assign associative unital algebras to local regions in an $n$-dimensional space in a manner that captures the locality structure of the underlying lattice theories. The net-theoretic perspective admits a natural categorification, in which the category of topological $n$-nets carries the structure of a Morita category. Moreover, the nets of greatest physical interest form a Karoubi-complete subcategory that shares many of the characteristic structural features of the category of finite TQFTs.
Our main result identifies the group of locality-preserving symmetries of (3+1)-dimensional quantum spin systems (i.e. quantum cellular automata modulo stabilization, finite-depth quantum circuits, and translations) with the full Witt group of nondegenerate braided fusion categories, thereby answering a question posed by Haah.
\end{abstract}

\tableofcontents

\section{Introduction}

Higher category theory has recently attracted considerable interest in physics, driven by the recognition that generalized symmetries in continuum quantum field theories are encoded by topological operators that have the structure of a higher category. In the case when the symmetry is finite, higher fusion categories have been central to the classification of symmetries and to the study of finite topological field theories (TQFTs) in higher dimensions \cite{Kong:2020wmn,johnson2022classification,JFY,Kong:2022cpy,Bhardwaj:2022yxj,Bhardwaj:2022lsg,Bhardwaj:2022maz,Bhardwaj:2024xcx,Bullimore:2024khm,Bhardwaj:2024qiv,Bhardwaj:2025piv,Teixeira:2025qsg,Decoppet:2025eic,Debray:2025kfg}.\footnote{What we mean by finite is that the TQFTs only have finitely many simple topological operators in every codimension.}

 Symmetries have played a central role in shaping our understanding of field theory, so naturally the study of symmetries of quantum systems formulated in discrete settings through algebras of local observables is also a rich and important subject. In condensed-matter physics such systems are commonly referred to as \textit{quantum spin systems}, see \cite{nachtergaele2004quantum} for a review. Since locality is intrinsic to the microscopic construction of these systems, it is natural to require the symmetry operations to also act in a local way as well. We define a \textit{quantum cellular automaton} to be a locality-preserving symmetry of a quantum spin system, which implements an automorphism of its algebra of observables in a way such that observables supported in any region are mapped to another region that is a uniformly bounded distance away. 

Just as the classification of symmetries is a fundamental problem in continuum theories, one may ask how QCAs are classified for low-dimensional quantum spin systems. This problem has garnered attention \cite{Freedman:2019ncn,Freedman:2019ucy,Haah:2018jdf} not least due to its overlap with category theory, homotopy theory \cite{Kapustin:2025nju,Kapustin:2025rhp,Tu:2025bqf,Czajka:2025mme}, topological phases \cite{Fidkowski:2019nju,Jones:2026gmd}, and lattice anomalies \cite{Kawagoe:2025ldx,Feng:2025yge,Shirley:2025yji}. For certain simple quantum spin systems whose algebra of observables is generated by the Pauli group, their locality-preserving symmetries are called Clifford QCAs and their classification has been studied through the lens of algebraic L-theory \cite{Haah:2019fqd,Haah:2022yyo,Geiko:2023zof,Yang:2025jvn}, with constructions of nontrivial  (3+1)d examples given in \cite{Sun:2025vta,Sun:2026toj}.

To classify the locality-preserving symmetries of general quantum spin systems, it is useful first to reformulate the problem in a manner guided by our understanding of symmetries in continuum theories.
In the study of an individual physical system, or even a purely mathematical object, one is naturally led to consider the ambient space parametrizing all such objects. Indeed, it is often difficult to define a single object without specifying the space or category of objects to which it belongs.
In physics, the most prominent example is the ``space of theories'' or ``space of symmetries'', although its precise mathematical incarnation varies across different settings.\footnote{See a review of this discussion in a continuum context in \cite[Section 1.1]{Stockall:2025ngz}.} With this global perspective in mind, the central question we address is:
\begin{question}
What is the appropriate category of quantum spin systems in arbitrary dimensions, together with the morphisms that capture the physically meaningful relations between them?
\end{question}

While the global perspective is necessarily more abstract and definition-driven, it allows us to identify the appropriate abstraction of quantum spin systems that is amenable to categorification. This, in turn, provides a systematic framework for studying higher-dimensional quantum spin systems and for giving a precise classification of their symmetries, leading to higher QCA groups. 
In light of this discussion, our goal is to address a fundamental homotopy-theoretic problem: the computation of the homotopy groups of the QCA spectrum. To achieve this, we bring together the theory of nets from algebraic quantum field theory and operator algebras with techniques from  category theory, specifically $\infty$-categories and categorified Karoubi completion aka condensation completion.

\subsection{Main Results}\label{subsection:mainresults}

We develop a category to describe quantum spin systems in $n$-dimensions, reframed in the language of topological $n$-nets. As in field theory, the most general category of quantum spin systems capable of accommodating the full range of physically relevant structures will be highly intricate. We do not claim that the framework developed here incorporates every structure one might wish to impose; in particular, questions regarding the unitary structure of  quantum systems are not addressed in this work.

Our results are two-fold. First, we construct a category of topological $n$-nets, which serves as the ambient space of theories that we wish to study. Second, we identify the distinguished subcategory of finite semisimple topological $n$-nets whose properties are controlled. We then realize the nets in this subcategory over a particular background space relevant for lattice models, thereby obtaining quantum spin systems whose symmetries can be analyzed using categorical structure.

\begin{rem}
   Our definition of a quantum spin system is motivated by the perspective commonly adopted in condensed-matter physics, in which the system is implicitly embedded in an ambient manifold, or merely a metric space. In particular, we do not regard such a lattice system as living solely on an abstract graph or simplicial space.
\end{rem}

\begin{notation}
    Unless  otherwise specified, ``$n$-dimensions'' will always reflect the spacetime dimension. The topological $n$-nets are hence also defined in $n$-spacetime dimensions. A lattice $\Z^d$ denotes a $d$-spacetime dimensional lattice.
\end{notation}

\begin{maintheorem}[{\cref{thm:Netn}}]
      There is a symmetric monoidal $(\infty,n+1)$-category $\Net_n$ with the following data for the nontrivial morphisms:
  \begin{itemize}
 \item Objects are topological $n$-nets,
        \item Tensor product and direct sum of two topological $n$-nets $\cA$ and $\cB$ exist and is given by tensor product and direct sum of the algebras assigned by the nets. The tensor unit is the trivial topological $n$-net $\underline{\mathbb{C}}$ and the zero object is $\underline{0}$.
         \item For $1\leq k<n$, a $k$-morphism  $\mathcal F:\cA \to \cB$ is a locality-preserving defect $n$-net of codimension-$k$.  The identity $k$-morphism is the restriction of $\cA$ to a defect $n$-net of codimension-$k$.
         \item An $n$-morphism $H:\cA \to \cB$ is  a sector of the partial net $\cF: \Disk^{n}_{n}\to \Alg(\Vect_\mathbb{C})$.  The identity $n$-morphism is the  vacuum sector $H_\cA$. The composition of two $n$-morphisms $H:\cA \to \cB$ and $K: \cB \to \cC$ is the tensor product vector space $K \otimes_\cB H$.
          \item An $(n+1)$-morphism is a homomorphism of sectors.
  \end{itemize}
\end{maintheorem}

The nets appearing in the preceding theorem may be regarded as abstract, cosheaf-like mathematical objects. They need not a priori possess any physical interpretation. The principal role of this generality is to provide a robust space of theories that accommodates assignments of arbitrary unital associative algebras. 
Those nets that define higher-dimensional quantum spin systems are expected to form a subcategory exhibiting many of the structural features characteristic of the category of topological orders.
In particular, the finiteness conditions imposed by their physical origin suggest that these systems should be modeled by objects in a higher fusion category.

Let $\mathbf{nVect}^{\fd}_{\mathbb C}$ denote the $n$-category of fully dualizable $n$-vector spaces, which may be regarded as a fusion $n$-category and as the Morita category of multifusion $(n-2)$-categories.
\begin{maintheorem}[{\cref{thm:QS}}]
    The category of $n$-dimensional quantum spin systems $\QS_n$ is a symmetric monoidal condensation complete $(n+1)$-category equivalent to $\npVect^{\fd}_\mathbb{C}$.
\end{maintheorem}

The morphisms in this category will classify locality-preserving symmetries of lattice systems in various spacetime dimensions, thanks to the looping property that this category enjoys. We will only be concerned with spacetime dimensions up to $4$ in this work, and we will not include symmetries implemented by translation. 

\begin{maintheorem}[{\cref{cor:QCAclassification}}]\label{mainthm:QCA}
   Let $1 \leq n\leq 3 $. The groups of QCAs  over $\mathbb{C}$ up to quantum circuits and stabilization on the spaces $\Z^{n+1}$ are classified by the following groups:
    \begin{itemize}
        \item For $n=3$, $\cQ(\Z^4)/\cC(\Z^4)= \Witt$: Morita classes of braided fusion 1-categories with trivial Müger center, i.e. the Witt group of nondegenerate braided fusion 1-categories \cite{davydov2010witt}. 
        \item For $n=2$, $\cQ(\Z^3)/\cC(\Z^3)= 0$: Morita classes of
        braided fusion 0-categories over $\C$ with trivial center.\footnote{A braided fusion 0-category over $\C$ is a commutative algebra over $\C$.}
        \item For $n=1$, $\cQ(\Z^2)/\cC(\Z^2)= 0$: Morita classes of central simple algebras over $\C$.
    \end{itemize}
\end{maintheorem}

Finally, we perform a natural extrapolation of our results to classify fermionic QCAs, which are automorphisms of local $\Z/2$-graded superalgebras. As we have not explicitly done the categorical constructions in the fermionic case, we state the following results as claims, but whose details can surely be filled in. 

\begin{clm}[{\cref{rem:fermionicQCA}}]
   Let $1 \leq n\leq 3 $. The group of \textbf{fermionic} QCAs  over $\mathbb{C}$ up to quantum circuits and stabilization on the space $\Z^{n+1}$ are classified by the following groups:
    \begin{itemize}
        \item For $n=3$, $\sWitt$: the Witt group of slightly degenerate braided fusion categories \cite{davydov2013structure}.
        \item For $n=2$, the trivial group: in agreement with the main result of \cite{Kwan:2026kcc}.
        \item For $n=1$, $\Z/2$:  in agreement with the classification result of \cite{fidkowski2019interacting}.
    \end{itemize}
\end{clm}

In the bulk of the paper we will prove that stabilization and quantum circuits have precise analogs in the context of the category $\npVect^{\fd}_\C$. This also supports the claim that the translation into the framework of nets is faithfully capturing all the features of quantum spin systems and their symmetries.

 \cref{cor:QCAclassification} abstractly establishes the existence of QCAs corresponding to classes beyond the ``pointed'' Witt group. However, I do not know how to construct a single representative (nor did I even try), due to technical deficiencies. I believe that understanding how to give a constructing for a representative is an illuminating direction for future work. 
So far there has been a body of work on realizing QCAs in the ``pointed'' Witt group in (3+1)d \cite{Shirley:2022lhu,Yang:2025jvn,YY:2026lnf}, and also in the presence of global symmetries \cite{Inamura:2026kco}. Because the Clifford group plays a central role in these realizations, a natural question can be raised as to whether classes beyond the pointed Witt group can be realized by transformations lying at higher levels of the Clifford hierarchy.


\begin{rem}
    From the perspective of the QCA spectrum, the groups computed here are precisely the ``negative'' homotopy groups of the spectrum $\mathbf Q(\pt)$. These groups are of primary interest when viewing QCA as a lattice invertible phase. A natural question for future work is how to formulate a framework in which negative $K$-theory groups recover the QCA groups described above.
\end{rem}

\subsection*{Outline}
This paper is structured as follows.
\begin{itemize}
    \item[\S\ref{section:quantumspin}]: We first review  quantum spin systems with a focus on their locality-preserving symmetries, namely QCAs.  In particular we introduce the QCA group which classifies QCA up to finite-depth circuit, and abstractions to the QCA space, QCA spectrum and explain their properties. This will set the stage for what the rest of the paper aims to recover and then ultimately extend using topological $n$-nets. 
    \item[\S\ref{section:Nets}]: We introduce topological $n$-nets and morphisms and higher morphisms implemented by defect $n$-nets of higher codimension. We show that this formalism can capture the equivalence condition used in the QCA group.
    \item[\S\ref{section:catofnets}]: We build the category $\Net_n$ of topological $n$-nets, whose morphisms are defect $n$-nets of all codimension.
    \item[\S\ref{section:condensation}]: We discuss the categorification of idempotents in higher categories to condensation monads. We emphasize the role of unitality in the theory of condensation monads and identify the key steps required to equip condensations in a category with a unital structure. 
    \item[\S\ref{section:catofQS}]: We  construct the category of quantum spin systems $\QS$ as a condensation complete category whose objects are nets that assign algebras satisfying appropriate finiteness conditions. We then show how this category is used to classify QCAs up to (3+1)d.
    \end{itemize}

\section{Quantum Spin Systems and Their Symmetries}\label{section:quantumspin}

This section gives a summary of the construction of the QCA group, space, and spectrum.
The reader whom is comfortable with the definitions are invited to move onto \S\ref{section:Nets}.

\begin{notation}\label{notation:X}
    For this section when we consider a space $X$, we will always assume that $X$ is a metric space. Furthermore we assume that  $X$ contains a countable, coarsely dense subspace that is a uniformly discrete, locally finite metric space of bounded geometry. These are the types of spaces of greatest physical interest for lattice models. 
\end{notation}

\subsection{Review of Quantum Spin Systems}

\begin{defn}
    Let $X$ be a space as in \cref{notation:X}. A function $q: X \to \mathbb{N}$ is called \textit{locally finite} if the space
    \begin{equation}
        \Lambda = \{x\in X| q_x>1\}
    \end{equation}
 has finite intersection with any  bounded subset of $X$.
    Let $\mathbb{N}^X_{\mathrm{lf}}$ be the partially ordered set of all locally finite functions, with $q\leq r$ if and only if $q_x | r_x$ for all $x \in X$.
\end{defn}

\begin{defn}\label{defn:spinsystem}
    A \textit{quantum spin system} over a space $X$ consists of a commutative ring $R$, and a locally finite function $q\in \mathbb{N}^X_{\mathrm{lf}}$. Such a system is determined by its \textit{algebra of local observables}, which is a product of matrix algebras: 
    \begin{equation}
        \cA(X,q) = \bigotimes_{x \in X} \mathrm{Mat}(R^{q_x}),
    \end{equation}
    where $\mathrm{Mat}(R^{q_x})$ is the ring of $q_x \times q_x$ matrices over $R$, and where the tensor product is taken over $R$.
\end{defn}

Given any two systems determined by $q$ and $q'$, one can perform a stacking operation to obtain a new system given by their pointwise product $q \cdot q'$.
Equivalences between quantum spin systems are implemented by isomorphisms that preserve locality in the algebra as follows.

\begin{defn}\label{def:autalpha}
    Given two $R$-quantum spin systems over a metric space $X$ with metric $\rho$, a locality-preserving homomorphism  is an $R$-algebra
homomorphism 
\begin{equation}
    \alpha: \cA(X,q)\to \cA(X,q')
\end{equation}
with the property that there exists $\ell >0$ such that for every $x\in X$ and every $A \in \cA(\{x\},q)$, we have 
\begin{equation}
    \alpha(A) \in \cA(D_\ell(x),q'), \quad D_\ell(x) := \{ y\in X|\, \rho(x,y) \leq \ell\}\,.
\end{equation}
\end{defn}

\noindent For most of the explicit results in this paper we work over $R=\C$, but it would be interesting to explore the analogs over different rings. We say that a quantum spin system defined by $q$ on $X$, equipped with a locality-preserving automorphism $\alpha$ is a \textit{quantum  cellular automaton} (QCA). 
QCAs thus provide the natural framework for describing symmetries of lattice systems.

\begin{example}
    Consider a quantum spin system defined on $\Z^2$. 
 The GNVW index \cite{Gross:2011yvb} measures the net flow of quantum information generated by a locality-preserving automorphism on such a system. It is defined locally, by comparing the operator algebras transported across a spatial cut, but is independent of the choice of cut and therefore gives a global invariant of the dynamics. For bosonic spin chains, this index is valued in  $\mathbb Q_{>0}$. Finite-depth local circuits have trivial index, while a translation of $d$-dimensional local degrees of freedom by one lattice site has index $d$, with the inverse translation having index $d^{-1}$. Thus the GNVW index may be viewed as measuring the translation-like component of a $(1+1)$-dimensional locality-preserving automorphism.
 
    As stated in \S\ref{subsection:mainresults}, we will not include these QCAs which arise from translation in our study of the group of QCAs in each dimension. One justification for this is suggested by the interpretation of QCAs as pumps for invertible phases. In this sense, a translation does not generate a new phase, but merely shifts the spatial position of an existing one.
\end{example}

We  dedicate the next subsection to  reviewing three important gadgets in the classification of QCAs: these being the QCA group, QCA space, and QCA spectrum introduced in \cite{Ji:2026fka}.
The properties of the QCA group, space, and spectrum, provide the primary motivation for this paper, as our goal is to recover them from the perspective of topological $n$-nets and ultimately provide new tools for their computation. The QCA conjecture provides the overarching organizational principle motivating the development of these structures, with each concept playing a distinct role in its formulation. We therefore begin by reviewing this conjectural framework.

\subsection{The QCA Conjecture}
Before stating the QCA conjecture, we find it illuminating to first review its conceptual origins in the theory of invertible gapped phases of matter. A gapped phase of matter is an equivalence class of gapped states on a $d$-dimensional lattice, where two states are considered equivalent if they can be connected by a continuous path along which the spectral gap remains open. Gapped phases can be equipped with the structure of a commutative monoidal under a stacking operation. Moreover, the trivial phase is the monoidal unit, and a gapped state is invertible if there is another gapped state such that their stacking is in the trivial phase.
For each $d$, the set of invertible gapped phases form an abelian group on $\Z^d$ given by $\cI(\Z^d)$.
Following a conjecture of Kitaev \cite{Kitaev2011,Kit13,Kit15,Kit19,Kit23}, there should exist a space $\mathbf{IP}(\Z^d)$ of invertible states on $\Z^d$ for  each $d$, such that $\pi_0(\mathbf{IP}(\Z^d)) \cong \cI(\Z^d)$. Furthermore, the sequence of spaces should form an $\Omega$-spectrum of invertible phases $\{\mathbf{IP}(\Z^d)\}_{d\geq 0} =\mathbb{IP}$. Kubota \cite{Kubota:2025qfl} has given a constructing of such a spectrum; we also direct the reader to Kubota for a more complete list of references regarding the spectrum of gapped phases.

Analogous to Kitaev’s conjectural picture for invertible phases, the QCA conjecture asserts that the spaces of QCAs in varying dimensions assemble into an $\Omega$-spectrum \cite{Long:2024kxl}. From this perspective, QCAs are closely related to invertible phases of matter and may themselves be regarded as their lattice counterparts.
In \cite{Ji:2026fka} the authors established a QCA spectrum involving a sequence of spaces of QCA on $\Z^n$ for $n\geq 0$, using constructions common in K-theory such as group completion. At present, explicit computations of the homotopy groups of the QCA spectrum are known only in low degrees. We therefore begin by reviewing the groups that have already been determined, as these provide the fundamental consistency checks that our framework must recover.

\subsubsection{The QCA Group}\label{subsubsection:qcagroup}
We begin by reviewing the QCA group, whose elements are stable equivalence classes of QCAs. We then introduce quantum circuits, a distinguished class of QCAs by which we quotient in order to obtain the desired classification.

\begin{defn}\label{def:stableequiv}
    Two QCA $[q,\alpha]$ and $[r,\beta]$ are \textit{stably equivalent} if there exists a $s$ with $q|s$ and $r|s$ such that $\iota_{q\to s}(\alpha) = \iota_{r\to s}(\beta)$, where $\iota$ is the embedding of one quantum spin system into another as a map of the matrix algebras at each site.
\end{defn}

\noindent This is also known as a blend equivalence in the literature \cite{Freedman:2019ncn,Tu:2025bqf,Czajka:2025mme}.

\begin{defn}
      Let $\cQ(X,q)$ be the group generated by the stable equivalence  classes of QCAs $[q,\alpha]$, for $\alpha$ as in \cref{def:autalpha}, with the product given by 
     \begin{equation}
         [q,\alpha] \cdot [r,\beta] = [s,\iota_{q\to s}(\alpha)  \circ \iota_{r\to s}(\beta)]
     \end{equation}
     for $q|s$ and $r|s$ and $s_x = \mathrm{lcm}(q_x,r_x)$.

     The \textit{total QCA group} is the colimit in the category of groups:
     \begin{equation}
        \cQ(X): = \varinjlim_{q \in \mathbb{N}^X_{\mathrm{lf}}} \cQ(X,q)\,.
     \end{equation}
 \end{defn}

 \begin{defn}
     Let $\cA(X,q)$ be a quantum spin system over the complex numbers. Suppose that $X = \coprod_j X_j$ where $X_j$ are  strict subsets of $X$, and let $q(X_j) = \prod_{x \in X_j}q_x$.
     A \textit{single layer quantum circuit}  on $\cA(X,q)$ is a choice of automorphism $\alpha_j \in \Aut_\C(\cA(X_j,q))$ on each $X_j$ which acts on $\cA(X_j,q)$ by conjugation via an unitary matrix.\footnote{Technically this will use the $\mathrm C^*$ property, but we will forget this information.}
     The product $\prod_j \alpha_j$ defines a locality-preserving automorphism of $\cA(X,q)$ which acts independently on the disjoint tensor factors of $\cA(X,q)$.

     Let $\cC(X,q)\subseteq \cQ(X,q)$ be the subgroup generated by single-layer circuits, which upon passing to the colimit in $q$ gives the group $\cC(X)\subseteq \cQ(X)$.
 \end{defn}

\noindent Our main results are concerned with understanding the (abelian group) $\cQ(X)/\cC(X)$ which isolates the topologically non-trivial behaviors of QCAs, as opposed to those that are circuits which are built up from local automorphisms. Since we are working over the complex numbers, we will consider circuits to be implemented by unitary matrices.

For $(1+1)$d quantum spin systems, the group $\cQ(\Z^2)/\cC(\Z^2)$ has been computed over a general ring $R$, and given by the Grothendieck group $K_0(\mathrm{Az}(R))$ where $\mathrm{Az}(R)$ is the category of Azumaya algebras over $R$ \cite[Theorem 1]{Ji:2026fka}. Recall that an $R$-algebra $A$ is an Azumaya algebra if there exists an $R$-algebra $B$ such that $A \otimes_R B \simeq \End_R(P)$ where $P$ is a faithfully projective $R$-module. We will see how the description of the group $\cQ(\Z^2)/\cC(\Z^2)$ can be recovered from our categorical model using topological $2$-nets in \cref{ex:azumaya}.

 \subsubsection{The Space and Spectrum of QCA}\label{subsubsection:spacespectrum}

As established in \cite{Ji:2026fka}, there is a space $\bfQ(\Z^n)$ for every $n \in \mathbb{N}$, with the property that $\pi_0(\bfQ(\Z^n))$ is the QCA group $\cQ(\Z^n)/\cC(\Z^n)$. The space is constructed via Segal K-theory applied to  a symmetric monoidal 1-category $\mathbf C(X)$ defined on a space $X$, whose objects are quantum spin systems, the monoidal structure is given by pointwise stacking,  and morphisms are locality-preserving isomorphisms.

The classifying space $\rB \mathbf C$ for the category $\mathbf{C}(X)$ is a homotopy commutative, homotopy associative $H$-space, with $\pi_0(\rB \mathbf{C})$ a commutative monoid.

The formalism of Segal K-theory implements a group completion $\mathrm{B} \mathbf{C} \to K(\mathbf C(X))$ where  $K(\mathbf C(X))$ is the \textit{K-theory space} associated to $\mathbf{C}(X)$, and is a group-like $H$-space. 
\begin{defn}
    The space of QCAs over $X$ is $\bfQ(X) := \Omega K (\mathbf C(X))$.
\end{defn}

Making use of Thomason's double mapping cylinder construction \cite{Thomason1979FirstQuadrant}, Ji and Yang furthermore show  \cite[Theorem E]{Ji:2026fka} that the sequence of spaces $\{\mathbf{Q}(\Z^n)\}_{n\geq 0}$ satisfies a looping property, giving rise to a homotopy equivalence
\begin{equation}
   \Omega  \bfQ(\Z^{n+1}) \simeq \bfQ (\Z^n)\,,
\end{equation}
giving us a spectrum of QCAs. Due to the fact that we have a $\Omega$-spectrum for QCAs, and that $\pi_i(\Omega \bfQ(\Z^{n+1})) \cong  \pi_{i+1}(\bfQ(\Z))$, we see that
the homotopy groups in degree $-n$ of $\bfQ(\pt)$ capture the groups of QCA in dimension $n$. It is these negative homotopy groups that we will capture in the rest of this paper. Probing negative homotopy groups is an appropriate formulation of the problem, since the structure we seek to capture using topological $n$-nets reflects ``categorical depth'' rather than homotopical height.

The truncated spectrum of QCAs constructed here does not arise from K-theory; instead, our construction draws on techniques developed in the theory of higher Morita categories \cite{Scheimbauer:2014zty,haugseng2017higher}. It would be interesting to recast this construction within a K-theoretic framework. Such a formulation, however, would necessarily involve negative K-theory, which is not directly accessible from a symmetric monoidal category through Segal K-theory alone.

\begin{rem}
    A version of the QCA conjecture defined on quantum spin systems built from C$^*$-algebras and unitary circuits is also established in \cite[Section 5]{Ji:2026fka}. In this case, QCAs are locality-preserving $*$-isomorphisms.
\end{rem}

\section{Quantum Spin Systems as Topological $n$-nets}\label{section:Nets}
To define topological $n$-nets, we work on a space  given by an $n$-dimensional cube $M=(0,1)^n$. This space is thought of as the background spacetime on which quantum spin systems are defined. In \S\ref{subsection:higherQCAgroups} we will apply this formalism defined for $M$ to the (discrete) space we actually care about for quantum lattice models.

Let $\Disk^{n}$ be the 1-category of $n$-dimensional disk regions of $M$,  whose morphisms are inclusions of disk regions. 
 We define a \textit{partial $n$-net} to be a functor $\cA: \mathbf{D} \to \Alg(\Vect_\mathbb{C})$,\footnote{Here, $\Alg(\Vect_\mathbb{C})$
 is the $(\infty,1)$-category of associative algebras over $\Vect_\C$, with morphisms given by algebra homomorphisms. There is another category of algebras one can consider which is the Morita category of finite semisimple algebras over $\mathbb C$. This is a 2-category and will be denoted $\tVect^{\fd}_\C$, which we take to be a fusion 2-category.} where $\mathbf{D}\subset \Disk^{n}$ is a subcategory. A partial $n$-net thus defines a family of algebras $\cA(I)$, along with a family of morphisms with respect to disk inclusion $\cA(I\subset J): \cA(I)\to \cA(J)$, with identity  $\cA(I \subset I) = \id_\cA$ and composition $\cA(J \subset K) \circ \cA(I \subset J) = \cA(I \subset K)$.

\begin{defn}
     A \textit{sector} of a partial $n$-net $\cA:\mathbf D \subset \Disk^{n} \to \Alg(\Vect_\mathbb{C})$ is a vector space $H_\cA$ equipped with a family of maps $\rho_I: \cA(I)\to \End(H_\cA)$ for $I \in \mathbf D$, such that for all $I\subset J \in \mathbf D$, we have $\rho_J\circ \cA(I\subset J) = \rho_I$. A homomorphism of sectors is a linear map $f$ that is a $\cA(I)$-module homorphism for each $I \in \mathbf{D}$.
\end{defn}

\begin{defn}
    A sector $H$ is \textit{irreducible} if it is
neither zero nor the direct sum of two nonzero sectors. Furthermore $H$ is \textit{semisimple} if it is a finite
direct sum of irreducible sectors.
\end{defn}

\begin{defn}
  A \textit{topological $n$-net} is a pair of a functor and a sector $(\cA,H_\cA)$, where  $\cA:\Disk^{n}\to \Alg(\Vect_\C)$, 
subject to the following properties for $I,J,K \in \Disk^{n}$:\footnote{The index $n$ is labeling spacetime dimension, following the notation from \S\ref{subsection:mainresults}.}
 \begin{itemize}
     \item If $I,J\subset K$ are disjoint, then the image of $\cA(I)$ and $\cA(J)$ in  $\cA(K)$ commute.
     \item If $I,J\subset K$ and the interiors of $I$ and $J$ cover $K$, then $\cA(K)$ is generated by the images of $\cA(I)$ and $\cA(J)$ under inclusion of regions.
     \item The direct sum of two nets is given by 
 $(\cA \oplus \cB)(I)=\cA(I)\oplus \cB(I)$ and the tensor product is $(\cA \otimes \cB)(I)=\cA(I)\otimes \cB(I)$\,.
 \end{itemize}
 A sector $H_\cA$ which is equivalent to $\cA$ as the regular $\cA\text{-}\cA^{\op}$-bimodule, is called the \textit{vacuum sector}.\footnote{If the manifold that the nets are defined on have a specific orientation, then the opposite algebra assigned by the net $\cA^{\op}$ can be thought of as the algebra on the manifold with opposite orientation.}  We will often denote $(\cA,H_\cA)$ by $\cA$ for simplicity.
\end{defn}

\begin{example} There are two fiducial nets one can consider. The first is the zero topological $n$-net $(\underline{0},H_{\underline{0}})$ is one where to any interval $I$, the net gives the trivial assignment $\underline{0}(I)=0$. The second is the constant topological $n$-net  $(\underline{\mathbb C},H_{\underline{\mathbb C}})$, where to any interval $I$, the net assigns $\underline{\mathbb{C}}(I) = \mathbb{C}$.
\end{example}

\begin{rem}
   Though we are greatly inspired by the work of Kong--Zhang \cite{Kong:2022hjj}, our perspective on topological $n$-nets will be different than that of \cite{Kong:2022hjj}, where the definition is motivated by a higher-dimensional generalization of conformal nets \cite{BDH1,BDH2,BDH3,BDH4}. In particular, their framework retains the action of the diffeomorphism group $\mathrm{Diff}(S^{n-1})$, whereas this additional geometric structure plays no role in the questions we consider and is therefore omitted. Likewise, rather than working with von Neumann algebras, we pass to their underlying algebraic structures. Since our primary interest lies in the higher-categorical structure of symmetries and defects, rather than the analytic aspects of algebraic quantum field theory, this algebraic setting captures the essential features relevant to our construction while avoiding the additional technicalities inherent to operator algebras and unitarity. In particular, defining finiteness for topological $n$-nets in our definition is much simpler. 
   
   In what comes later, our definition of the category of nets also differs from the definition of discrete net given by Jones in \cite{Jones:2023ptg}, where morphisms between nets were defined to be bounded spread \textit{isomorphisms} from the outset. We will also not impose isotony in the definition of topological $n$-nets, which is the statement that $\cA(I)$ is a subalgebra of $\cA(J)$ if $I \subset J$.
\end{rem}

\begin{example}\label{ex:spinsystem}
     A quantum spin system on a $(n-1)$-dimensional spatial lattice as in Definition \ref{defn:spinsystem}, and one dimension of time,  can be described as a topological $n$-net. In particular we can take the algebra on a region $I$ to be given by $\cA(I) = \bigotimes_{x\in I^\delta} \mathrm{Mat}(\C^{q_x})$, where the tensor is only over a discrete subset of points in $I^\delta \subset I$, reflecting the lattice nature of the quantum spin system. Let $V_x = \mathbb C^{q_x}$, and consider $V_I = \bigotimes_{x \in I^\delta} V_x$. Then we take the vacuum sector to be $H_{\cA(I)} = V_I \otimes V^*_I$. 
     The pair manifestly satisfies the definition for a topological $n$-net.
 \end{example}

 For any topological $n$-net $\cA$ and sector $H_\cA$, consider a closed region $R \subset (0,1)^n$. Define an algebra on $H_\cA$ denoted by $\mathcal{O}_\cA(H,R)\subset \End(H)$, to be an algebra generated by the images of $\cA(I)$ for all $I \subset R$. 
The following definitions will be particularly important when formulating the category of  quantum spin systems using topological nets.

\begin{defn}\label{def:semisimpleNet}
      A topological $n$-net $(\cA,H_\cA)$ is \textit{reducible} if the sector $H_\cA$ is reducible as a $\mathcal O_\cA(H_\cA,(0,1)^n)\text{-}\mathcal O_\cA(H_\cA,(0,1)^n)^{\op}$-bimodule. Furthermore, $(\cA,H_\cA)$ is
      \textit{semisimple} if $H_\cA$ is semisimple as a  $\mathcal O_\cA(H_\cA, (0,1)^n)\text{-}\mathcal O_\cA(H_\cA, (0,1)^n)^{\op}$-bimodule.
\end{defn}

\begin{defn}\label{def:finiteNet}
    Let $I,J \in \Disk^n$ be disjoint disk regions, and let $I'$ and $J'$ denote their complements. A topological $n$-net $(\cA,H_\cA)$ is \textit{finite} if it is semisimple and $H_\cA$ is dualizable as a  $\cO_\cA(H_\cA,I \cup J)\text{-}\cO_\cA(H_\cA,I' \cap J')^{\text{op}}$-bimodule for all $I$ and $J$.
\end{defn}

\subsection{Defect $n$-nets}\label{subsection:defectnets}

A topological $n$-net in essence behaves like a locally constant copresheaf. To proceed with categorification we need to describe morphisms between these nets, so we introduce defect $n$-nets. Their definition requires a stratification of the underlying manifold $M$ on which the nets are supported. A stratification of $M$ consists of a collection of closed dimension-$i$ submanifolds $M_i$ such that
\begin{equation}
    M = M_n \supset M_{n-1} \supset \ldots \supset M_0 \supset M_{-1} = \emptyset\,,
\end{equation}
around which $M_{i+1}$ is equipped with a local parametrization. We will take $M = (0,1)^n$, with the following type of stratification.\footnote{One can straightforwardly generalize the stratifications in this section to $n$-dimensional lattices that are infinite in each direction.}

\begin{defn}
    The \textit{standard affine flag of index $k$} is given by the stratification:
    \begin{equation}
        (0,1)^{n}\supset (0,1)^{n-1}\times \{\tfrac{1}{2}\} \supset \ldots \supset (0,1)^{n-k+1}\times \underset{k-1\text{-entries}}{\underbrace{\{(\tfrac{1}{2},\ldots,\tfrac{1}{2})\}}}\supset (0,1)^{n-k}\times \underset{k\text{-entries}}{\underbrace{\{(\tfrac{1}{2},\ldots,\tfrac{1}{2})\}}}\,.
    \end{equation}
    A local picture of the affine flag which contains stratifications up to codimension-$k$ is called a \textit{standard stratified disk of index $k$}.
\end{defn}

\begin{example}
    For $n=2$ we consider the following standard stratified disks of index $k$ labeled below.
    \begin{equation}\label{eq:2disks}
   \begin{tikzpicture}
\tikzset{
  boundary/.style={draw=frame, line width=0.9pt},
  stratum/.style={draw=frame, line width=0.9pt, line cap=round},
  point/.style={circle, fill=frame, inner sep=1.2pt}
}

\begin{scope}
  \begin{scope}
    \clip (0,0) circle (0.5);
    \fill[c0] (-1,-1) rectangle (1,1);
  \end{scope}
  \draw[boundary,white] (0,0) circle (0.5);
  \node[font=\small] at (0,-0.78) {$k=0$};
\end{scope}

\begin{scope}[xshift=2.4cm]
  \begin{scope}
    \clip (0,0) circle (0.5);
    \fill[c0] (-1,-1) rectangle (0,1);
    \fill[c3] (0,-1) rectangle (1,1);
  \end{scope}
  \draw[boundary,white] (0,0) circle (0.5);
  \draw[stratum] (0,-0.45) -- (0,0.45);
  \node[font=\small] at (0,-0.78) {$k=1$};
\end{scope}

\begin{scope}[xshift=4.8cm]
  \begin{scope}
    \clip (0,0) circle (0.5);
    \fill[c0] (-1,-1) rectangle (0,1);
    \fill[c3] (0,-1) rectangle (1,1);
  \end{scope}
  \draw[boundary,white] (0,0) circle (0.5);
  \draw[stratum] (0,-0.45) -- (0,0.45);
  \node[point] at (0,0) {};
  \node[font=\small] at (0,-0.78) {$k=2$};
\end{scope}
   \end{tikzpicture}
    \end{equation}
   The blue and red regions show that connected components of the 2-dimensional strata, on which a topological 2-net lives. The lines of higher codimension can helpfully be conceptualized as local regions of defects embedded into the space. 
\end{example}

In general, we are interested in stratifications of $M=(0,1)^n$ where for each point $x\in M$, there is an open neighborhood of $x$ that is diffeomorphic as a stratified space to a standard stratified disk of some index. Let $\Disk^{n}_k$ be a  subcategory of $\Disk^{n}$ consisting of stratified  $n$-disks of index $k$, i.e. $n$-disks that are diffeomorphic to the standard stratified $n$-disk. 

\begin{defn}
    A \textit{defect $n$-net of codimension $k$} for $1\leq k \leq n$ is a pair of a partial topological $n$-net and a sector $(\cF,H_\cF)$, where $\cF: \Disk^{n}_k\to \Alg(\Vect_\C)$, subject to the following conditions for $I,J,K \in \Disk^{n}_k$:
     \begin{itemize}
     \item If $I,J\subset K$ are disjoint, then $[\cF(I),\cF(J)]=0$ as subalgebras of $\cF(K)$.
     \item If $I,J\subset K$ and the interiors of $I$ and $J$ cover $K$, then $\cF(K)$ is generated by the images of $\cF(I)$ and $\cF(J)$ under inclusion of regions.
 \end{itemize}
\end{defn}

We now discuss how defect $n$-nets of different codimension   implement morphisms in the category of topological $n$-nets $\Net_n$. Showing that composition of morphisms at each level closes, as well as the coherence relations, will be relegated to \S\ref{section:catofnets}.

\begin{example}
   To begin discussing morphisms, let $\cF$ be a defect 1-net of codimension-1, also known as a sector. By definition $\cF$ is defined on standard disk regions of index 1 on $(0,1)$, i.e. on $\{s\}$. One can view $\cF$ as a 1-morphism from the topological 1-nets on each side of $s$ shown below
\begin{equation}
     \begin{tikzpicture}[x=4cm,y=1cm,baseline]
  \coordinate (A) at (0,0);
  \coordinate (B) at (1,0);
  \coordinate (s) at (0.43,0);

  \draw[leftblue, line width=3pt, line cap=round] (A) -- (s);
  \draw[rightred, line width=3pt, line cap=round] (s) -- (B);

  \fill (s) circle (1.8pt);

  \node[below=4pt] at (A) {$0$};
  \node[below=4pt] at (B) {$1$};
  \node[above=4pt] at (s) {$s$};
\end{tikzpicture}
\end{equation}
Let us consider the inclusions of the blue and red intervals into the whole interval. 
The restriction of $\cF$ along the inclusion to  the blue region is the source topological 1-net, and the restriction on the red region is the target topological 1-net.
  Consider another defect 1-net of codimension-1 given by $\mathcal G$, defined on a stratification given by $t$ as below
  \begin{equation}
     \begin{tikzpicture}[x=4cm,y=1cm,baseline]
  \coordinate (A) at (0,0);
  \coordinate (B) at (1,0);
  \coordinate (t) at (0.7,0);

  \draw[leftblue, line width=3pt, line cap=round] (A) -- (t);
  \draw[rightred, line width=3pt, line cap=round] (t) -- (B);

  \fill (t) circle (1.8pt);

  \node[below=4pt] at (A) {$0$};
  \node[below=4pt] at (B) {$1$};
  \node[above=4pt] at (t) {$t$};
\end{tikzpicture}
\end{equation}
   The 2-morphisms are maps between two  defect $1$-nets of codimension-$1$, $\cF$ and $\mathcal{G}$ (with $\cF$ and $\cG$ defined by the strata at $s$ $t$, respectively.), which is a linear map of vector spaces from $\cF$ to $\cG$ that takes $s$ to $t$.
\end{example}

\begin{example}
    We now consider morphisms of topological 2-nets. As before, the defect 2-nets of codimension-1 defined on a stratified disk given by \cref{eq:2disks} middle,  implement 1-morphisms. In addition, the 2-morphisms are implemented by defect 2-nets of codimension-2 $\cF$ (which are by definition sectors), defined on a stratified disk given by \cref{eq:2disks} right. This is done as follows. Consider splitting the disk along a horizontal line passing through the middle point, and reshaping the regions above and below the line into stratified  disks of index 1 as shown below.

\begin{equation}
    \begin{tikzpicture}
    \tikzset{
  boundary/.style={draw=frame, line width=0.9pt},
  stratum/.style={draw=frame, line width=0.9pt, line cap=round},
  point/.style={circle, fill=frame, inner sep=1.2pt}
}
        \begin{scope}[xshift=4.8cm]
  \begin{scope}
    \clip (0,0) circle (0.5);
    \fill[c0] (-1,-1) rectangle (0,1);
    \fill[c3] (0,-1) rectangle (1,1);
  \end{scope}
  \draw[boundary,white] (0,0) circle (0.5);
  \draw[stratum] (0,-0.45) -- (0,0.45);
  \node[point] at (0,0) {};
  \node [rotate=45, opacity=5] at (1,.8) {$\longmapsto$};
  \node [rotate=-45, opacity=5] at (1,-.8) {$\longmapsto$};
   \clip (0+2,0+1.5) circle (0.5);
    \fill[c0] (-1+2,-1+1.5) rectangle (0+2,1+1.5);
    \fill[c3] (0+2,-1+1.5) rectangle (1+2,1+1.5);
   \draw[boundary,white] (0+2,0+1.5) circle (0.5);
  \draw[stratum] (0+2,-0.45+1.5) -- (0+2,0.45+1.5);
\end{scope}
\begin{scope}
     \clip (0+6.85,0-1.5) circle (0.5);
    \fill[c0] (-1+6.85,-1-1.5) rectangle (0+6.85,1-1.5);
    \fill[c3] (0+6.85,-1-1.5) rectangle (1+6.85,1-1.5);
     \draw[boundary,white] (0+6.85,0-1.5) circle (0.5);
  \draw[stratum] (0+6.85,-0.45-1.5) -- (0+6.85,0.45-1.5);
\end{scope}
    \end{tikzpicture}
\end{equation}
Here, the source and target of the 2-morphisms can be taken to be either restricting $\cF$ to the bottom or top disk.
     Finally, 3-morphisms are implemented as linear maps between sectors.
\end{example}

\subsection{Justification for Topological $n$-nets}\label{subsection:justification}
Let us take for granted the main results in $\S\ref{section:catofnets}$ and $\S\ref{section:catofQS}$, which establish the category $\Net_n$ of topological $n$-nets and the category of quantum spin systems $\QS_n$ as a subcategory.
We will show that by considering quantum spin systems from the perspective of topological $n$-nets, the 1-morphisms of  $\Net_n$ given by defect $n$-nets of codimension-1, capture the properties and classification of QCAs studied on the space $\Z^n$ via the QCA group. 

We first show that the equivalence condition on invertible 1-morphisms in $\Net_n$ can equivalently be thought of as stable equivalence of QCAs as in \cref{def:stableequiv}.

\begin{prop}\label{prop:bimoduleequiv}
  Let $\cF$ and $\cG$ be two 1-morphisms in $\Net_n$. $\cF$ is stably equivalent to $\cG$ if and only if there exists $\cD:\cF\to \cG$ which is an invertible 2-morphism.
\end{prop}

\begin{proof}
   Suppose that $\cF$ and $\cG$ are stably equivalent. This means they are equivalent upon tensoring with the identity for both 1-morphisms. This is the same as the two being separated by an invertible defect $n$-net $\cD$ that evaluates to $\cF$ and the identity outside of the support of $\cF$, as well as $\cG$ and the identity outside of $\cG$. Conversely, suppose that $\cD$ is an invertible 2-morphism. 
   Then by the definition of a defect net of codimension-2,
   we can consider restricting $\cD$ to either $\cF\otimes \mathbf 1$ or $\mathbf 1 \otimes \cG$, and consider the compositions $\cD^{-1}\circ \cD: \cF \otimes \mathbf 1 \to \cF \otimes \mathbf 1$ and $\cD\circ \cD^{-1}:\mathbf 1 \otimes \cG \to \mathbf 1 \otimes \cG$, where the two composites carry  $\cF$ or $\cG$ into itself after  tensoring on by $\mathbf 1$. The fact that $\cD\circ \cD^{-1} = \cD^{-1}\circ \cD = \id$ then gives the desired stable equivalence.
\end{proof}

\begin{prop}\label{prop:circuit}
   A quantum circuit is a 1-morphism in $\Net_n$ that is stably equivalent to the identity 1-morphism. 
\end{prop}

\begin{proof}
    Consider a quantum circuit $\alpha_J$ that acts on the  algebra in a region $J$ via conjugation by an unitary $U_J$. Denote this circuit by $\alpha_J = \mathrm{Ad}_{U_J}$, which is technically a single layer circuit. We will show the associated 1-morphism in $\Net_n$, which has the interpretation as a bimodule, is stably equivalent to the identity. In particular from $\alpha_J$  we can form the bimodule $_{\cA(J)} \cA(J)_{\cA(J)^{\alpha_J}}$, where the right action is twisted by $\alpha_J$, so that for $m\in \cA(K)$ the right action is given by 
    \begin{equation}
        m \cdot b = m \alpha_J(b) = m U_J b U^{-1}_J\,.
    \end{equation}
     We show that there is an isomorphism 
    $$T_{U_J}: _{\cA(J)} \cA(J)_{\cA(J)^{\alpha_J}}\to _{\cA(J)} \cA(J)_{\cA(J)}\,,$$
    implemented by $T_{U_J}(m) = m U_J$. We can check that $T_{U_J}$ respects the left $\cA$-action
    \begin{equation}
        T_{U_J}(a \cdot m) = a m U_J = a \cdot T_{U_J}(m)\,,
    \end{equation}
    as well as the twisted right $\cA$-action
    \begin{equation}
        T_{U_J}(m \cdot b) =  T_{U_J}(m U_J b U^{-1}_J) = m U_J b U^{-1}_J U_J = T_{U_J} (m) \cdot b\,.
    \end{equation}
    Combining the two actions we see that for all $a,m,b \in \cA$,
    \begin{equation}
       T_{U_J} (a\cdot m \cdot b) =  a \cdot T_{U_J}(m) \cdot b\,,
    \end{equation}
    and therefore $T_{U_J}$ is an $\cA\text{-}\cA$-bimodule map. Moreover, one can check that $T_{U_J}$ has an inverse given by 
$$T_{U^{-1}_J}: _{\cA(J)} \cA(J)_{\cA(J)}\to _{\cA(J)} \cA(J)_{\cA(J)^{\alpha_J}}$$
which acts as $T_{U^{-1}_J}(m) = m U^{-1}_J$. Therefore $T_{U_J}$ is an $\cA\text{-}\cA$-bimodule isomorphism, and $\alpha_J$ can be made isomorphic to the identity automorphism $\id_J$.

Now suppose we consider the system $\cA(J) \otimes \mathcal{T}$ where $\cT$ is an ancilla system placed on some region outside of $J$ and within some large region $K$ on which the topological $n$-nets are defined. Then the stabilized circuit is  $\alpha_J \otimes \id_\cT  = \Ad_{\widetilde{U}}$ where $\widetilde{U}$ acts on a region of $J$ as well as the region where $\cT$ is defined. We can consider an analogous bimodule map 
\begin{equation}
    T_{\widetilde{U}}: _{\cA(J)\otimes \cT } \cA(J)\otimes \cT_{(\cA(J)\otimes \cT)^{\alpha_J\otimes \id_{\cT}}} \to _{\cA(J)\otimes \cT} \cA(J)\otimes \cT_{\cA(J)\otimes \cT}\,,
\end{equation}
and in the same way as above show that it is an isomorphism.  Therefore, $\alpha_J \otimes \id_\cT \simeq \id_{J}\otimes \id_\cT \simeq  \id_{K}$ and we obtain a stable equivalence of the twisted bimodule to the trivial bimodule i.e. a stable trivialization of the bimodule.

In the case when the circuit is finite-depth,  we can view each layer as  implemented by a product of unitaries over the entire space. A single layer of the circuit is the tensor product over disjoint regions of such local operations given by $U_J$. Since the
supports are disjoint, one  can consider local stable trivializations as discussed above that are mutually compatible within the layer. Therefore the associated  1-morphism for the single layer circuit is
built by tensor products of local automorphisms that are each stably isomorphic to the identity. Hence, the entire layer is stably equivalent to the identity 1-morphism.
A finite-depth
circuit is then a finite composite of such layers. 
\end{proof}

A key feature to notice is that while the automorphism $\alpha=\mathrm{Ad}_{U_J}$ can be nontrivial as an automorphism of $\cA$, when viewed as a 1-morphism i.e. a defect net, the twisted bimodule $_\cA \cA_{\cA^\alpha}$ is stably equivalent to the identity bimodule. Conceptually, 
inner automorphisms are invisible after passing from automorphisms to Morita-theoretic defects that implement the map.  The preceding results show that the group of invertible 1-morphisms of $\Net_n$ that are not equivalent to the identity has an equivalent description to QCAs up to stabilization and up to quantum circuits; i.e. what is captured by the group $\cQ(X)/ \cC(X)$ in \S\ref{subsubsection:qcagroup}.

\begin{example}\label{ex:azumaya}
  Let us   apply Propositions \ref{prop:bimoduleequiv} and \ref{prop:circuit}  to quantum spin systems in $(1+1)$d, which are objects in $\QS_2$. In this case the category $\QS_2$ is a 3-category, and the group of QCA modulo stabilization and quantum circuits is given by the group of invertible 1-morphisms therein, which is equivalent to  $\pi_0((\tVect^\fd_{\mathbb{C}})^\times)$. This group is equivalent to the Morita classes of Azumaya algebras over $\C$, and matches the description of the group of QCAs in $(1+1)$d given in \cite[Theorem 1]{Ji:2026fka}. 
\end{example}

\noindent Motivated by this last example, the group of higher-dimensional QCAs should be equivalent to the group of ``higher Azumaya algebras''. This was also pointed out in \cite{ludewig2026quantum} in the context of coarse homology. Making this realization precise will be done in \S\ref{section:catofQS}.

\section{The Category of Topological $n$-nets}\label{section:catofnets}
In this section we consider the category of topological $n$-nets with its defect $n$-nets of codimension-$k$ as morphisms, and show that it forms a symmetric monoidal $(\infty,n+1)$-category $\Net_n$.  We are inspired by the techniques in the factorization algebra literature, in particular involving constructable factorization algebras \cite{gwilliam2018duals}. In the context of the topological nets that we have defined, there are similarities between the two fields. When one assigns analytic objects such as von Neumann algebras as in the case of conformal nets, the differences become more stark.

\begin{thm}\label{thm:Netn}
  There is a symmetric monoidal $(\infty,n+1)$-category $\Net_n$ with the following data for the nontrivial morphisms:
  \begin{itemize}
 \item Objects are topological $n$-nets,
        \item Tensor product and direct sum of two topological $n$-nets $\cA$ and $\cB$ exist and is given by tensor product and direct sum of the algebras assigned by the nets. The tensor unit is the trivial topological $n$-net $\underline{\mathbb{C}}$ and the zero object is $\underline{0}$.
         \item For $1\leq k<n$, a $k$-morphism  $\mathcal F:\cA \to \cB$ is a locality-preserving defect $n$-net of codimension-$k$.  The identity $k$-morphism is the restriction of $\cA$ to a defect $n$-net of codimension-$k$.
         \item An $n$-morphism $H:\cA \to \cB$ is  a sector of the partial net $\cF: \Disk^{n}_{n}\to \Alg(\Vect_\mathbb{C})$.  The identity $n$-morphism is the  vacuum sector $H_\cA$. The composition of two $n$-morphisms $H:\cA \to \cB$ and $K: \cB \to \cC$ is the tensor product vector space $K \otimes_\cB H$.
          \item An $(n+1)$-morphism is a homomorphism of sectors.
  \end{itemize}
\end{thm}

We will give the proof of this theorem in \S\ref{subsection:catnet}, after introducing some  background in the next couple of subsections.

\subsection{Compositions of Defect $n$-nets}\label{subsection:compositions}
We will work up to describing morphisms in $\Net_n$, starting off with $n=1$ and then generalizing to higher dimensions.
 1-morphisms in $\Net_1$ are defect $1$-nets of codimension-1 i.e. sectors, and composition of 1-morphisms is given by composition of sectors. Suppose that $H_\cD:\cA \to \cB$ is a defect 1-net of codimension-1 defined on a manifold $X_1$ and implements a morphism from a 1-net on the blue to  a 1-net on the green region in \cref{eq:compose1morph}, and $H_\cE:\cB \to \cC$ is a defect 1-net of codimension-1 defined on $X_2$ which implements a morphism from a 1-net defined on the green to a 1-net defined on the red region:
\begin{equation}\label{eq:compose1morph}
     \begin{tikzpicture}[x=4cm,y=1cm,baseline]
     \begin{scope}
  \coordinate (A) at (0-1,0);
  \coordinate (B) at (1-1,0);
  \coordinate (s) at (0.5-1,0);

  \draw[leftblue, line width=3pt, line cap=round] (A) -- (s);
  \draw[midgreen, line width=3pt, line cap=round] (s) -- (B);

  \fill (s) circle (1.8pt);

  \node[below=4pt] at (A) {$0$};
  \node[below=4pt] at (B) {$1$};
  \node[above=4pt] at (s) {$s_1$};
       \end{scope}
       \hspace{20mm}
            \begin{scope}
                 \coordinate (A) at (0,0);
  \coordinate (B) at (1,0);
  \coordinate (s) at (0.5,0);

  \draw[midgreen, line width=3pt, line cap=round] (A) -- (s);
  \draw[rightred, line width=3pt, line cap=round] (s) -- (B);

  \fill (s) circle (1.8pt);

  \node[below=4pt] at (A) {$0$};
  \node[below=4pt] at (B) {$1$};
  \node[above=4pt] at (s) {$s_2$};
            \end{scope}
\end{tikzpicture}
\end{equation}
   The composition is defined by joining the manifolds together along the green region to form the manifold $(X_1\backslash (s_1,1)) \cup Y \cup  (X_2\backslash (0,s_2))$, and rescaling the manifold to be back to the interval $(0,1)$ as follows:
    \begin{equation}
       \begin{tikzpicture}[x=9cm,y=1cm,baseline]
  \coordinate (A) at (0,0);
  \coordinate (B) at (1,0);
  \coordinate (sone) at (0.34,0);
  \coordinate (stwo) at (0.68,0);

  \draw[leftblue, line width=3pt, line cap=round] (A) -- (sone);
  \draw[midgreen, line width=3pt, line cap=round] (sone) -- (stwo);
  \draw[rightred, line width=3pt, line cap=round] (stwo) -- (B);


  \fill (sone) circle (1.8pt);
  \fill (stwo) circle (1.8pt);

  \node[below=4pt] at (A) {$0$};
  \node[below=4pt] at (B) {$1$\,.};
  \node[above=4pt] at (sone) {$s_1$};
  \node[above=4pt] at (stwo) {$s_2$};
\end{tikzpicture}
   \end{equation}
The sectors $H_\cD$ and $H_\cE$ compose to give another sector $H_\cF := H_\cE \boxtimes_{\cB} H_\cD: \cA \to \cC$, where by $\cB$ we mean the algebra that $\cB$ assigns to the common region.  The restrictions of $\cF$ to the blue and red regions give $\cA(I)$ and $\cC(J)$.
Thus the action of $\cF$ on $H_\cF$ is via $H_\cD$ and $H_{\cE}$ when $\cF$ is restricted to its source and target. 
There is a subtle difference between a pair of composable 1-morphisms, and their composition. In particular we require that $\cD$ and $\cE$ are equivalent on the intermediate region shaded in green, and hence there is a unique (up to coherent isomorphism) way to define the entire defect 1-net of codimension-1. We demand that all defect $n$-nets have the property that when composing them, they are equivalent on their common source/target. Defect $n$-nets that have this property satisfy a condition which we call \textit{collar excision}, over the glued underlying manifold. This condition is the net-theoretic formulation of $\otimes$-excision for factorization algebras, as defined in \cite[Definition 3.15]{ayala2015factorization}. Its significance lies in ensuring that the theory of nets is capable of encoding the principle of locality, and is an essential feature for typical physical theories and in particular quantum spin systems. 
\begin{defn}\label{def:collarexcision}
    Let $M$ be a manifold such that $M\cong M_1 \cup_Y M_2$, where $Y \subset M$ is a common submanifold shared by 
    two submanifolds $M_1$ and $M_2$ of $M$ whose union recovers all of $M$.  A defect $n$-net $\cF$ satisfies \textit{collar excision} if the morphism
    \begin{equation}
        \cF(M_1 \cup_Y M_2) \xrightarrow{\cong} \cF(M_1) \otimes_{\cF(Y)} \cF(M_2)
    \end{equation}
    is an equivalence in $\Alg(\Vect_\mathbb{C})$, for every choice of $M_1$, $M_2$, and $Y\subset M$ such that $ M\cong M_1 \cup_Y M_2$.
\end{defn}

Henceforth, we assume that defect $n$-nets of every codimension satisfy collar excision; equivalently, one may incorporate this condition as an axiom in the definition of defect $n$-nets when studying  physically relevant applications.

\begin{example}
    Consider the quantum spin system given in \cref{ex:spinsystem} with three adjacent regions $I,J,K$ as in \cref{eq:compose1morph}. The vacuum sector for the first defect 1-net is $ V_I \otimes V^*_J$
and the vacuum sector for the second is $V_J \otimes V^*_K$. Sewing these two defects over the intermediate region $J$ means taking the relative tensor product over the algebra $B(J)\cong \End(V_J)$, which gives 
\begin{equation}
    (V_I \otimes V^*_J) \otimes_{V_J \otimes V^*_J} (V_J \otimes V^*_K)\,.
\end{equation}
The contraction over the intermediate matrix algebra eliminates $V^*_J \otimes_{V_J \otimes V^*_J} V_J\cong \C$. The remainder is the vacuum sector for the combined cut $I \cup J$.
\end{example}
We summarize the discussion thus far in the following Lemma.
\begin{lem}\label{lem:sectorcompose}
   Let $H_\cD: \cA \to \cB$ and $H_\cE:\cB \to \cC$ be two sectors in  $\Net_n$. Then the composition is a well defined sector $H_{\cF}: \cA \to \cC$. Furthermore, if $H_\cD$ and $H_\cE$ are semisimple, then $H_\cF$ is semisimple.
\end{lem}


\begin{rem}
  We give a nod to constructable factorization algebras, and explain why composition is considerably more straightforward in that context. The key point is that a constructible factorization algebra already encodes its multi-input colimit structure as part of its defining data. Consequently, the pushforward along a stratified map is computed simply by evaluating the factorization algebra on the preimage. In contrast, a net records only the local algebras assigned to regions. To define composition, one must first assemble the local observables associated to the preimage into a single algebra by taking the algebra generated by those observables on the fused state space. Thus, whereas composition for factorization algebras is built into the factorization structure itself, composition for nets requires an additional reconstruction step at the level of local algebras.
\end{rem}

We now explain how composition works for defects $n$-nets in lower codimension.
Let $\cD: \cA \to \cB$ and $\cE: \cB \to \cC$  be defect $n$-nets of codimension-$k$. In other words, $k$-morphisms in the category $\Net_n$. We will now explain how to construct the composition $\cF = \cE \circ \cD:\cA \to \cC$ as another defect $n$-net of codimension-$k$. Suppose that $\cD$ and $\cE$ are defined on  two stratified manifold $X_1=(0,1)^n$ and $X_2=(0,1)^n$. We glue the manifolds together along the submanifolds $(\tfrac{1}{2},1)\times (0,1)^{n-1}\subset X_1$ and $(0,\tfrac{1}{2})\times (0,1)^{n-1}\subset X_2$, to form the manifold 
$$X = \left( X_1 \backslash (\tfrac{1}{2},1)\times (0,1)^{n-1} \right) \cup Y \cup \left(X_2 \backslash (0,\tfrac{1}{2})\times (0,1)^{n-1} \right)\,,$$ on which $\cF$ is defined. The region $Y$ plays the same role as in the example earlier, which is the common region where the two manifolds are joined.

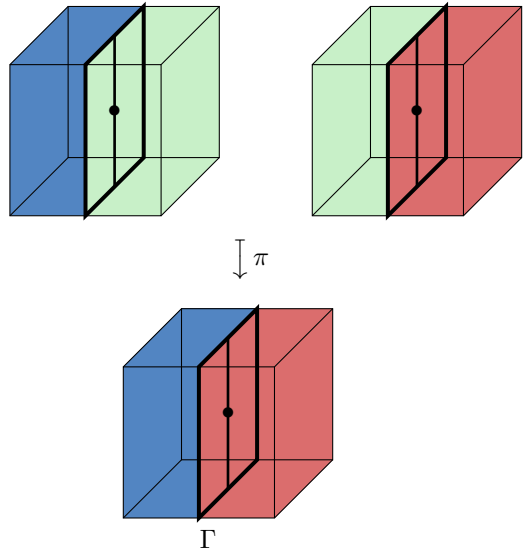
\begin{figure}[ht]
    \centering
   \begin{tikzpicture}
\centering
 \fill[leftblue]
(0-3,0,0) -- (1-3,0,0) -- (1-3,2,0) -- (0-3,2,0) -- cycle;

\fill[leftblue]
(0-3,2,0) -- (1-3,2,0) -- (1-3,2,2) -- (0-3,2,2) -- cycle;

\fill[leftblue]
(0-3,0,0) -- (0-3,2,0) -- (0-3,2,2) -- (0-3,0,2) -- cycle;
\fill[leftblue]
(0-3,0,0) -- (1-3,0,0) -- (1-3,0,2) -- (0-3,0,2) -- cycle;

\fill[lightgreen] (-2,2,0) -- (-1,2,0) -- (-1,2,2) -- (-2,2,2) -- cycle;

\fill[lightgreen] (-1,0,0) -- (-1,2,0) -- (-1,2,2) -- (-1,0,2) -- cycle;
\fill[lightgreen] (-2,0,2) -- (-2,2,2)-- (-2,2,0)--(-2,0,0)--cycle;
\fill[lightgreen] (-2,2,0) -- (-1,2,0) -- (-1,0,2) -- (-2,0,2) -- cycle;
    \draw (0-3,0,0) -- (2-3,0,0) -- (2-3,2,0) -- (0-3,2,0) -- cycle;
    \draw (0-3,0,2) -- (2-3,0,2) -- (2-3,2,2) -- (0-3,2,2) -- cycle;
    \draw[line width=1.5pt] (-2,0,2) -- (-2,2,2)-- (-2,2,0)--(-2,0,0)--cycle;
    \draw[line width=1pt] (-2,2,1)-- (-2,0,1);
     \node at (-2,1,1) {$\bullet$};
    \draw (0-3,0,0) -- (0-3,0,2);
    \draw (2-3,0,0) -- (2-3,0,2);
    \draw (2-3,2,0) -- (2-3,2,2);
    \draw (0-3,2,0) -- (0-3,2,2);
    \hspace{1cm}
    \fill[lightgreen]
(0,0,0) -- (1,0,0) -- (1,2,0) -- (0,2,0) -- cycle;

\fill[lightgreen]
(0,2,0) -- (1,2,0) -- (1,2,2) -- (0,2,2) -- cycle;

\fill[lightgreen]
(0,0,0) -- (0,2,0) -- (0,2,2) -- (0,0,2) -- cycle;
\fill[lightgreen]
(0,0,0) -- (1,0,0) -- (1,0,2) -- (0,0,2) -- cycle;
\fill[rightred]
(0+1,2,0) -- (1+1,2,0) -- (1+1,2,2) -- (0+1,2,2) -- cycle;

\fill[rightred]
(0+2,0,0) -- (0+2,2,0) -- (0+2,2,2) -- (0+2,0,2) -- cycle;
\fill[rightred]
(0+1,0,0) -- (1+1,0,0) -- (1+1,0,2) -- (0+1,0,2) -- cycle;
\fill[rightred]
(0+1,0,2) -- (1+1,0,2) -- (1+1,2,2) -- (0+1,2,2) -- cycle;

     \draw (0,0,0) -- (2,0,0) -- (2,2,0) -- (0,2,0) -- cycle;
    \draw (0,0,2) -- (2,0,2) -- (2,2,2) -- (0,2,2) -- cycle;
     \draw[line width=1.5pt] (-2+3,0,2) -- (-2+3,2,2)-- (-2+3,2,0)--(-2+3,0,0)--cycle;
    \draw[line width=1pt] (-2+3,2,1)-- (-2+3,0,1);
     \node at (-2+3,1,1) {$\bullet$};
    \draw (0,0,0) -- (0,0,2);
    \draw (2,0,0) -- (2,0,2);
    \draw (2,2,0) -- (2,2,2);
    \draw (0,2,0) -- (0,2,2);
%
  \fill[leftblue]
(0-3+.5,0-3-1,0) -- (1-3+.5,0-3-1,0) -- (1-3+.5,2-3-1,0) -- (0-3+.5,2-3-1,0) -- cycle;

\fill[leftblue]
(0-3+.5,2-3-1,0) -- (1-3+.5,2-3-1,0) -- (1-3+.5,2-3-1,2) -- (0-3+.5,2-3-1,2) -- cycle;

\fill[leftblue]
(0-3+.5,0-3-1,0) -- (0-3+.5,2-3-1,0) -- (0-3+.5,2-3-1,2) -- (0-3+.5,0-3-1,2) -- cycle;
\fill[leftblue]
(0-3+.5,0-3-1,0) -- (1-3+.5,0-3-1,0) -- (1-3+.5,0-3-1,2) -- (0-3+.5,0-3-1,2) -- cycle;
   \fill[rightred]
(0-1.5,2-3-1,0) -- (1-1.5,2-3-1,0) -- (1-1.5,2-3-1,2) -- (0-1.5,2-3-1,2) -- cycle;

\fill[rightred]
(0+2-1-1.5,0-3-1,0) -- (0+2-1-1.5,2-3-1,0) -- (0+2-1-1.5,2-3-1,2) -- (0+2-1-1.5,0-3-1,2) -- cycle;
\fill[rightred]
(0+1-1-1.5,0-3-1,0) -- (1+1-1-1.5,0-3-1,0) -- (1+1-1-1.5,0-3-1,2) -- (0+1-1-1.5,0-3-1,2) -- cycle;
\fill[rightred]
(0+1-1-1.5,0-3-1,2) -- (1+1-1-1.5,0-3-1,2) -- (1+1-1-1.5,2-3-1,2) -- (0+1-1-1.5,2-3-1,2) -- cycle;

     \draw (0-1-1.5,0-3-1,0) -- (2-1-1.5,0-3-1,0) -- (2-1-1.5,2-3-1,0) -- (0-1-1.5,2-3-1,0) -- cycle;
    \draw (0-1-1.5,0-3-1,2) -- (2-1-1.5,0-3-1,2) -- (2-1-1.5,2-3-1,2) -- (0-1-1.5,2-3-1,2) -- cycle;
     \draw[line width=1.5pt] (-2+3-1-1.5,0-3-1,2) -- (-2+3-1-1.5,2-3-1,2)-- (-2+3-1-1.5,2-3-1,0)--(-2+3-1-1.5,0-3-1,0)--cycle;
    \draw[line width=1pt] (-2+3-1-1.5,2-3-1,1)-- (-2+3-1-1.5,0-3-1,1);
     \node at (-2+3-1-1.5,1-3-1,1) {$\bullet$};
      \node at (-2+3-1-1.5+.2,1-5.2,2.2) {$\Gamma$};
       \node[rotate=-90] at (-2+3-1-1.5+.6,-.5,2.2) {$\longmapsto$};
              \node at (-2+3-1-1.5+.9,-.5,2.2) {$\pi$};
    \draw (0-1-1.5,0-3-1,0) -- (0-1-1.5,0-3-1,2);
    \draw (2-1-1.5,0-3-1,0) -- (2-1-1.5,0-3-1,2);
    \draw (2-1-1.5,2-3-1,0) -- (2-1-1.5,2-3-1,2);
    \draw (0-1-1.5,2-3-1,0) -- (0-1-1.5,2-3-1,2);
   
\end{tikzpicture}
\caption{The two cubes give a three-dimensional representation of the manifolds $X_1=(0,1)^n$ and $X_2 = (0,1)^n$.
The green regions on the left and right represent the submanifolds  $(\tfrac{1}{2},1)\times (0,1)^{n-1}$ and $(0,\tfrac{1}{2})\times (0,1)^{n-1}$, respectively. It is also where the two manifolds are sewn together, and where we give an identification of the two nets in that region. The submanifold $\Gamma$ which forms from collapsing the green region is stratified by higher codimensional submanifolds; we only show strata given by a line and a point for simplicity.}
\label{fig:higherdimension}
\end{figure}

We define a map $\pi: X \to (0,1)^n$ that collapses the intermediate region $Y$ and, via a diffeomorphism of the manifold, identifies the resulting space with $(0,1)^n$ while preserving the stratification initially set on $X_1$ and $X_2$. To define $\cF$ we first need to define the composite vacuum sector, which is given by the relative tensor product, 
\begin{equation}
    H_\cF = H_\cE \boxtimes_\cB H_\cD\,.
\end{equation}

\begin{defn}\label{def:compositeF}
    The composite net $\cF: \cA \to \cC$ makes the following assignments for a disk region $I$ which intersects strata of codimension-$k$:
    \begin{itemize}
        \item if $I$ is entirely on the left of  the codimension-1 stratum $\Gamma$ in \cref{fig:higherdimension}, $\cF(I) = \cA(I)$\,,
        \item if $I$ is entirely on the right of  the codimension-1 stratum $\Gamma$, $\cF(I)=\ \cC(I)$\,,
        \item if $I$ intersects the codimension-1 stratum $\Gamma$, $\cF(I) = \mathcal{O}(H_\cF,\pi^{-1}(I))$\,.
    \end{itemize}
    More explicitly, in the final case, the associated algebra is generated by the algebras assigned to pre-collapsed regions  $J \subseteq \pi^{-1}(I)$.
\end{defn}

\begin{prop}\label{prop:composenets}
Let $\cD$ and $\cE$ be defect $n$-nets of codimension-$k$, with $\cF = \cE \circ \cD$ and sector $H_\cF$. The pair $(\cF,H_\cF)$ is a well defined defect $n$-net of codimension-$k$ in $\Net_n$.
\end{prop}
\begin{proof}
   It is clear that $\cF$ has the properties of a defect $n$-net of codimension-$k$.  While the sewing construction gives $H_\cF$ an action by $\cF$, it does not automatically give an identification with $\cF$. We hence need to show that as a $\cF\text{-}\cF^{\op}$-bimodule $H_\cF$ is the regular one, in order for $H_\cF$ to satisfy the vacuum property.\footnote{In this notation, we mean that it is a bimodule for the algebra which $\cF$ assigns to the source region (the left region in the diagrams) and the (opposite) algebra which $\cF$ assigns to the target region (the right region in the diagrams).} Let $V$ be a $\cF\text{-}\cF^{\op}$-bimodule, and define the set 
   $$Z(V,\cF) = \{ x \in V|\, ax = xa, \quad \forall a\in \cF\}\,.$$
   There is a isomorphism, natural in $V$, $\hom_{\cF\text{-}\cF^{\op}}(\cF,V)\cong Z(V,\cF)$, where each $f:\cF \to V$ is determined by an element $f(1_\cF)=x$, with the property that 
   \begin{equation}
       a f(1_\cF) =  f(a) = f(a \cdot 1_\cF)= a x 
   \end{equation}
   and 
   \begin{equation}
       f(a) = f(1_\cF \cdot a) = f(1_\cF) a = x a\,.
   \end{equation}
   Hence, $f(1_\cF)$ satisfies the property of an element in $Z(V,\cF)$.  Conversely, if $x\in Z(V,\cF)$ then define $f_x:\cF \to V$, with $f_x(a) = a x = x a$.
   We now relate $\cF$ and $H_\cF$, via the map  $\phi: \cF \to H_\cF$ that records the algebra action on $H_\cF$. For any (algebra) action of $\cF\text{-}\cF^\op$ on $\cF$, we use $\phi$ to transport that action to $H_\cF$, which can be done by \cref{def:compositeF}, third bullet. For $g \in \hom_{\cF\text{-}\cF^{\op}}(H_\cF,V)$ we consider $x=g(\phi(1_\cF))$, which is central with respect to every $a \in \cF$ by the relation
   \begin{equation}
       a g(\phi (1_\cF)) = g(a \phi (1_\cF)) = g(\phi (a)) =  g( \phi (1_\cF)a) =  g( \phi(1_\cF))a\,.
   \end{equation}
   We thus have an isomorphism, natural in $V$,  $\hom_{\cF \text{-}\cF^{\op}}(H_\cF,V)\cong Z(V,\cF)$. By using the fact that $$ \hom_{\cF \text{-}\cF^{\op}}(\cF,V)\cong Z(V,\cF)\cong \hom_{\cF \text{-}\cF^{\op}}(H_\cF,V),$$ along with the Yoneda lemma, we find that $\phi$ induces an isomorphism of bimodules and hence we establish the vacuum property of $H_\cF$. 
\end{proof}

\subsection{Complete $n$-fold Segal Spaces}\label{subsubsection:segal}
Although it is not always necessary to work with a model of $(\infty,n)$-categories, adopting a model can be helpful for proving some technical theorems. For the subject of topological $n$-nets and for defining $\Net_n$, we will use complete $n$-fold Segal spaces for our model, following \cite{lurie2008classification}. We start by reviewing simplicial spaces, which can be used to capture the idea of a $(\infty,1)$-category when a certain ``Segal condition'' as well as completeness is satisfied. We then step up to $n$-fold simplicial spaces that satisfy an $n$-fold Segal condition to model $(\infty,n)$-categories.  The necessity of this formalism is to capture relations that hold up to specified equivalence, and also the hierarchy of higher coherences governing those equivalences.

\begin{defn}
    Let $\Delta$ be the category of \textit{combinatorial simplices},  with objects the linearly ordered sets $[n] = \{ 0 < 1 < \ldots < n\}$. Given a pair of integers $m,n\geq 0$, $\hom_{\Delta}([m],[n])$ is given by the set of nonstrictly increasing maps $f:{0 <1 < \ldots <m}\to {0 <1 < \ldots <n}$.
\end{defn}

A  simplicial object in a $(\infty,1)$-category $\cC$ is a functor $X_\bullet: \Delta^\op  \to \cC$. A \textit{simplicial space} is a functor $X_\bullet: \Delta^\op  \to \mathbf{Spaces}$, and $X_n$ is the value of $X_\bullet$ when evaluated on the object $[n]$. The category $\mathbf{Spaces}$ has objects which are spaces, and by this we mean Kan simplicial sets.

\begin{defn}
    A simplicial space $X_\bullet$ is a (1-fold) \textit{Segal space} if for every pair of integers $m,n\geq 0$, the diagram
    \begin{equation}\label{eq:segalcondition}
        \begin{tikzcd}
            X_{m+n} \arrow[r] \arrow[d] & X_m \arrow[d]\\
            X_n \arrow[r]& X_0
        \end{tikzcd}
    \end{equation}
    is a homotopy pullback square. 
\end{defn}

The fact that the square in \cref{eq:segalcondition} is a pullback imposes the so-called \textit{Segal condition}, which is essential for modeling a $(\infty,1)$-category via a simplicial space. Specifically, for a composition of any length $k\geq 1$, one can make it equivalent to a sequence of $k$ composable 1-morphisms 
\begin{equation}\label{eq:lengthk}
    X_k \xrightarrow{\cong} X_1 \overset{h}{\times}_{X_0} X_1 \overset{h}{\times}_{X_0} \ldots \overset{h}{\times}_{X_0} X_1\,,
\end{equation}
where  the target of a 1-morphism is equivalent to the source of the next 1-morphism. One should view $X_0$ as the space of objects in the category $X$.

From a Segal space $X_\bullet$ we extract what should be thought of as an $(\infty,1)$-category $\cC$ in the following manner. 
\begin{itemize}
    \item The objects of $\cC$ are the points of the topological space $X_0$.
    \item The mapping space $\hom_\cC(x,y)$ between two points $x,y \in X_0$ is defined to be the homotopy fiber product $\{x\} \overset{h}{\times}_{X_0} X_1 \overset{h}{\times}_{X_0} \{y\}$.
    \item The composition law $\hom_\cC(x,y) \times \hom_\cC(y,z) \to \hom_\cC(x,z)$ for three points $x,y,z \in X_0$ is given by 
    \begin{align}\label{eq:compose}
        (\{x\} \overset{h}{\times}_{X_0} X_1 \overset{h}{\times}_{X_0} \{y\})\times (\{y\} \overset{h}{\times}_{X_0} X_1 \overset{h}{\times}_{X_0} \{z\}) &\to \{x\} \overset{h}{\times}_{X_0} X_1 \overset{h}{\times}_{X_0} X_1 \overset{h}{\times}_{X_0} \{z\} \\\notag 
        & \simeq  \{x\} \overset{h}{\times}_{X_0} X_2 \overset{h}{\times}_{X_0} \{z\}\,.
    \end{align}
    In the second step, we are using the Segal space assumption on $X_\bullet$ in order to obtain the homotopy equivalence.
    \item The data of the simplicial space $X_\bullet$ along with the Segal condition for other values of $n$ in \cref{eq:segalcondition} ensure that the  composition law in \cref{eq:compose} is associative up to higher coherent homotopy.
\end{itemize}

One can also attempt to extract a simplicial space from an $(\infty,1)$-category $\cC$. However, only enforcing the Segal condition on the simplicial space does not reflect the correct space of objects up to isomorphism in  $\cC$.
We will thus introduce the idea of completeness for Segal spaces, so that complete Segal spaces are in exact correspondence with $(\infty,1)$-categories. To this end, it is important to emphasize the fundamental shift from individual morphisms between objects to spaces of morphisms, organized by connected components.

\begin{defn}
    Let $X = X_\bullet$ be a Segal space. 
    The \textit{homotopy category of $X$} is denoted $h_1(X)$, with objects given by the set of objects in the space $X_0$, with the set of morphisms between two objects $x,y$ given by 
    \begin{equation}
        \hom_{h_1(X)}(x,y) = \pi_0(\hom_X(x,y)) = \pi_0\left( \{x\} \overset{h}{\times}_{X_0} X_1 \overset{h}{\times}_{X_0} \{y\}\right),
    \end{equation}
    and composition of morphisms is obtained by applying $\pi_0$ to the maps in \cref{eq:compose}.
\end{defn}

\begin{defn}
 Let $X = X_\bullet$ be a Segal space and $X^{\mathrm{eq}}_1 \subset X_1$, the subspace of those 1-morphisms that are invertible in $h_1(X)$.
    Then $X$ is \textit{complete} if there is an equivalence $s:X_0 \xrightarrow{\cong} X^{\mathrm{eq}}_1$, induced by the
degeneracy map $s$.\footnote{The degeneracy map is induced by the unique non-decreasing functor $\{0,1\}\to \{0\}$.}
\end{defn}

\begin{rem}
   A Segal space has the correct mapping spaces and compositions one would like for an $(\infty,1)$-category. However, one subtlety to recall is that a  $(\infty,1)$-category does not have a set of objects but  a space of objects and equivalences.
     Completeness ensures that if $x,y \in X_0$ are equivalent objects (reflected via isomorphisms in the category $X$), then the equivalence between them is already reflected as paths in the space of objects. 
\end{rem}

We will take the definition of a $(\infty,1)$-category to be a Segal space that is also complete. We now proceed to describe higher categories, via $n$-fold Segal spaces.
\begin{defn}
    An \textit{$n$-fold simplicial space} is a functor 
    \begin{equation}
        \Delta^\op \times \ldots \times \Delta^\op \to \mathbf{Spaces}
    \end{equation}
    where the product has $n$-factors. We denote by $X_{k_1,\ldots,k_n}$ the value of such a functor on $[k_1],\ldots,[k_n]$. 
\end{defn}

  A map of simplicial spaces $X\to Y$ is a weak homotopy equivalence if the induced map $X_{k_1,\ldots,k_n} \to Y_{k_1,\ldots,k_n}$ is a homotopy equivalence for every sequence of integers $k_1,\ldots,k_n \geq 0$.

\begin{defn}
    A \textit{complete $n$-uple Segal space} is a $n$-fold simplicial space such that for each $1\leq i \leq n$, and every $k_1,\ldots,k_{i-1},k_{i+1},\ldots,k_n\geq 0$, 
    \begin{equation}
        X_{k_1,\ldots, k_{i-1},\bullet,k_{i+1},\ldots,k_n}
    \end{equation}
    is a complete Segal space.
\end{defn}

The essence of this definition is that a complete $n$-uple Segal spaces are determined by complete Segal spaces in each of the $n$ possible directions.

\begin{defn}
    A  $n$-fold simplicial space $X$ is \textit{essentially constant} if there exists a weak homotopy equivalence of $n$-fold simplicial spaces $X'\to X$, where $X'$ is a constant functor.\footnote{Thought of as a functor, $X$ is constant if and only if for every sequence $k_1,\ldots,k_n \geq 0$, the canonical map $X_{0,\ldots,0}\to X_{ k_1,\ldots,k_n}$ is a weak homotopy equivalence. One should view $X_{0,\ldots,0}$ as the constant $n$-fold simplicial space, and in this case $X$ is weakly equivalent to the constant $n$-fold simplicial space associated to  $X_{0,\ldots,0}$.}  
\end{defn}

Given $X$ is a $n$-fold simplicial space, one can regard $X_{\vec{\bullet}}$ as a constant object of the category of $(n+1)$-fold simplicial spaces. This is done by declaring $X_{k_1,\ldots,k_n,k_{n+1}} = X_{k_1,\ldots,k_n}$. It is thus convenient to write $X_{k_1,\ldots,k_i}$ for $X_{k_1,\ldots,k_i,0,\ldots,0}$.
Furthermore, the essentially constant condition allows us to do is assume that $k_1,k_2,\ldots,k_i$ is a sequence of strictly positive integers, with the convention that the empty sequence is allowed, so that $X_{\emptyset} = X_0 = X_{0,\ldots,0}$ is the space of objects. 

\begin{defn}
    An \textit{$n$-fold Segal space} is a complete $n$-uple Segal space $X$ such that for every $1 \leq i \leq n-1$, $X_{k_1,\ldots, k_{i-1},0,\bullet,\ldots,\bullet}$ is essentially constant, i.e. 
    \begin{equation}
      X_{k_1,\ldots, k_{i-1},0,0,\ldots,0}  \rightarrow X_{k_1,\ldots, k_{i-1},0,k_{i+1},\ldots,k_{n}}
    \end{equation}
    is an equivalence. 
\end{defn}

Without the essentially constant condition, the complete $n$-uple Segal spaces behave like higher versions of double categories, which are categories with two notions of 1-morphisms: one being vertical and the other horizontal.  Thus this condition is used to cut down the category to only having a single type of morphism.
Analogous to how we defined $(\infty,1)$-categories, we can reasonably model $(\infty,n)$-categories as $n$-fold Segal spaces that are complete, and we will take this as the working definition in this paper.

\subsection{Constructing the Category $\Net_n$}\label{subsection:catnet}
We have defined the objects of $\Net_n$ as topological $n$-nets and composition of $k$-morphisms for $k\leq n$ as defect $n$-nets on a stratified manifold. When $k=n+1$, the morphisms are given by linear maps of sectors. 
We assemble all of the morphisms into the gadget $\Net_n$, which we show is a complete $(n+1)$-fold Segal space, giving it the properties of an $(\infty,n+1)$-category.

We recall that the topological $n$-nets and defect $n$-nets all take values in the $(\infty,1)$-category $\Alg(\Vect_\mathbb{C})$. To streamline the arguments that follow, we regard a topological $n$-net through the algebraic data it assigns to regions, and also for defect $n$-nets of all codimensions. Beginning with defect $n$-nets of codimension-1 we recover the algebras associated with topological $n$-nets, as well as maps between topological $n$-nets, and a composition structure of these maps along one dimension. Furthermore, higher-codimensional defect $n$-nets give rise to maps between lower-codimensional defect $n$-nets, with an additional direction of composition added for each level of codimension, as shown in \cref{prop:composenets}. Taken together, these observations indicate that a topological $n$-net should be thought of as algebras with $n$ compatible directions of composition. From this perspective, it is natural to view $\Net_n$ as a category that models the Morita category $\Alg_n(\Alg(\Vect_{\mathbb{C}}))$. Here, the additional $+1$ denotes the inherent 1-categorical structure of $\Alg(\Vect_{\mathbb{C}})$. Thus, to establish our theorem, we need to show that  $\Alg_n(\Alg(\Vect_{\mathbb{C}}))$ has the structure of a $(\infty,n+1)$-category.
We now introduce the right definitions and background to make this precise. 

\begin{defn}\label{def:Algn}
    The Morita category of $E_n$-algebras valued in  a symmetric monoidal $(\infty,N)$-category $\mathcal{C}$, denoted $\Alg_n(\cC)$,  is a $(\infty,n+N)$-category with objects $E_n$-algebras, 1-morphisms are bimodules of $E_{n-1}$-algebras, 2-morphisms are bimodules of bimodules of $E_{n-2}$ algebras, etc, and the morphism at level $n < k \leq n+N$ are given by morphisms in $\cC$.
\end{defn}

\begin{rem}
    In the usual Morita $(\infty,n+1)$-category where the objects are  $E_n$-algebras, the 1-morphisms are given by bimodules of $E_{n-1}$ algebras, etc, and finally a $(n+1)$-morphism is a bimodule map.  The defect $n$-nets of different codimension in $\Net_n$ also inherit the structure of bimodules and higher bimodules when viewed through the underlying algebras they assign. This is consistent with the perspective advocated in \S\ref{subsection:justification}, where the net-based approach was motivated by its ability to capture the properties of QCA through bimodule structures.
\end{rem}

We follow the line of reasoning in \cite[Section 3.3]{gwilliam2018duals} to construct $\Net_n$ using complete $n$-fold Segal spaces introduced in \S\ref{subsubsection:segal}. We start with some general considerations to build up the necessary language. From a general $(\infty,N)$-category $\cC$, one can extract diagrams of $(\infty,1)$-categories, which are the $k$-morphisms as in \cite{JFS}. This is done through considering strict higher $k$-categories $\Theta^{\vec k}$ acting as probe objects, with the property  that $[\Theta^{\vec k},\cC]$ extracts the space of $k$-morphisms of $\cC$; here $[A,B]$ denotes the mapping space from $A$ to $B$.

\begin{defn}
    An \textit{$n$-computad} $\Theta$ is a presentation of a strict $n$-category freely generated by sets of $k$-morphisms for $0\leq k \leq n$.\footnote{A strict $n$-category can be understood from its nerve, which is a $n$-fold simplicial set. The nerve of a strict $n$-category satisfies the Segal condition in each direction, and essentially constancy conditions. Thus the nerves of strict $n$-categories are $n$-fold Segal spaces. They are, however, rarely complete. If completeness is also a feature of the nerve then it is the case that all isomorphisms of the strict $n$-category are identities.}
    
    Associated to $\Theta$, the $(n-1)$-computad $\partial \Theta$ is the $(n-1)$-skeleton of $\Theta$. If $1\leq n$ then for each generating $n$-morphism there is a pair of parallel $(n-1)$-morphisms, called its \textit{source} and \textit{target}, in the strict $(n-1)$-category generated by $\partial \Theta$. For $k\leq n$, the $k$-skeleton of an $n$-computad $\Theta$ is the $k$-computad $\partial^{n-k} \Theta$.
\end{defn}

 One can consider $\Theta$ to be the full subcategory of strict $n$-categories on those $n$-categories that are freely generated by pasting diagrams of $n$-globes, which was proposed by Joyal \cite{Joyal1997Disks}.

\begin{defn}
    Consider the vector $\vec k = (\underset{k\text{-times}}{\underbrace{1,\ldots,1}},0,\ldots,0)$.
    The \textit{walking $k$-morphism} $\Theta^{\vec k}$ is the $k$-computad $\Theta^{\vec k}$ with $(k-1)$-skeleton given by $\Theta^{\vec{(k-1)}} \cup _{\partial \Theta^{\vec{(k-1)}}}\Theta^{\vec{(k-1)}}$, and a unique generating $k$-morphism with source  the unique generating $(k-1)$-morphism in the first $\Theta^{\vec{(k-1)}}$ and target in the unique generating $(k-1)$-morphism in the second $\Theta^{\vec{(k-1)}}$.
\end{defn}

\begin{example}
    The walking $0$-morphism $\Theta^{\vec{0}}$ is just the one element set $\{ \bullet \}$, while the walking $1$-morphism $\Theta^{\vec{1}}$ is the diagram $\{ \bullet \rightarrow \bullet \}$.
\end{example}
   
For $\cC$ a $(\infty,N)$-category we obtain a diagram of $(\infty,1)$-categories that extracts the $k$-morphisms by considering
\begin{equation}\label{eq:strongmorph}
 \Delta^N \ni \vec{k} \mapsto  \cC^\square_{\vec k} = \tau_{(\infty,1)}[\Theta^{\vec k},\cC]\,.
\end{equation}

\begin{rem}
 Given a $(\infty,n)$-category $\cC$ one can consider the category of vertical and horizontal morphisms between objects in $\cC$. A morphism between vertical and horizontal morphisms is called a lax and oplax natural transformation respectively. 
   One could also just consider strong natural transformations. In our consideration, we will always be considering strong natural transformations. The usage of the $\square$ notation in \cref{eq:strongmorph} reflects the fact that in \cite{JFS}, the authors studied these three types of natural transformations separately. The differences will not be important in this paper.
\end{rem}

Before stating one of the important structural results that follows, we introduce some terminology. 

\begin{defn}
    A category $\mathcal I$ is \textit{sifted} colimits indexed by $\mathcal I$ commute with finite products in $\mathbf{Set}$.  A \textit{sifted colimit} is a colimit of a diagram $F:\mathcal I \to \cC$ where $\mathcal I$ is a sifted category.
\end{defn}

\begin{defn}
    A category $\cC$ is \textit{$\otimes$-sifted-cocomplete} if it admits all sifted colimits, and for every object $X\in \cC$ tensoring on the left or on the right by $X$ preserves sifted colimits. A functor between $\otimes$-sifted-cocomplete categories is \textit{$\otimes$-sifted-cocontinuous} if it is symmetric monoidal and preserve sifted colimits. 
\end{defn}

This is the monoidal analogue of requiring that tensor product is compatible with the allowed colimits one can take. 

\begin{example}
    If a category has all (small) colimits then it  has all sifted colimits, so it is sifted cocomplete e.g. the category of sets, topological spaces, and simplicial sets. The category of nonempty sets is sifted-cocomplete but not cocomplete. In particular, the empty set is the initial element of the category of sets, and removing it spoils cocomplete since there is no colimit for the empty diagram.
\end{example}

\begin{defn}\label{def:sifted}
    Let $\cC$ be a symmetric monoidal $(\infty,N)$-category. The diagram $\cC^\square$ is \textit{$\otimes$-sifted-cocomplete} if it is a $N$-fold simplicial diagram of $\otimes$-sifted cocomplete categories and $\otimes$-sifted-cocontinuous functors.
\end{defn}

\begin{thm}[{\cite[Theorem 3.16]{gwilliam2018duals}}]\label{thm:buildingAlg}
    Let $\cC$ be a symmetric monoidal $(\infty,N)$-category such that $\cC^\square$ is $\otimes$-sifted-cocomplete. Then:
    \begin{enumerate}
        \item The $(n+N)$-tuple simplicial space $\Alg_n(\cC^\square_{\vec \bullet})_{\vec \bullet}$ satisfies the Segal condition and is complete in each variable separately.
        \item The symmetric monoidal structure on $\cC$ determines a symmetric monoidal structure on $\Alg_n(\cC)$.
    \end{enumerate}
\end{thm}

We consider $\Alg_n(\cC^\square_{\vec \bullet})_{\vec \bullet}$  as a $(\infty,n)$-by-$(\infty,N)$ double category $(\vec k, \vec l)\mapsto \Alg_n(\cC^\square_{\vec l})_{\vec k}$. The category of topological $n$-nets $\Net_n$ is the underlying $(\infty,n+1)$-category $\Alg_n(\cC)$, when $\cC = \Alg(\Vect_\C)$.

\begin{rem}\label{rem:netexcision}
    While nets do not behave exactly like factorization algebras in terms of their composition, we have shown in \cref{lem:sectorcompose} and \cref{prop:composenets} that we can nonetheless recover another defect $n$-net from the composition. In particular, in order to satisfy the Segal conditions we must impose that the nets satisfy collar excision given in \cref{def:collarexcision}, when gluing on overlapping regions. Then we can apply the tools of \cite{gwilliam2018duals} to build the Morita category from the structure of the nets and algebras they assign.
\end{rem}

With this setup in place, we can now prove the main result of this section. 
\begin{proof}[Proof of \cref{thm:Netn}]
    Set $\cC = \Alg(\Vect_\C)$, so that $N=1$ in \cref{def:sifted}. In this case the condition of $\cC^\square$ being $\otimes$-sifted-cocomplete translates to just requiring the category $\cC$ is  $\otimes$-sifted cocomplete. The compatibility with $\otimes$ is required for $\Alg_n(\cC)$ to have a symmetric monoidal structure. To show that the category $\Alg(\Vect_\C)$ is $\otimes$-sifted-cocomplete, consider the forgetful functor 
    \begin{equation}
        \mathrm{oblv}:\Alg(\Vect_\C) \to \Vect_\C\,,
    \end{equation}
    with left adjoint $F$ given by freely tensoring. This adjunction is monadic for the monad $T= \mathrm{oblv}\circ F$.  The tensor product in $\Vect_\mathbb{C}$ preserves colimits separately in each variable, and therefore given a sifted diagram $B_i$, for $i\in I$, in $\Alg(\Vect_\C)$, the colimits of associative algebras are computed on the underlying vector spaces as
\begin{equation}
    \mathrm{oblv}\left( \underset{i\in I}{\colim}\, B_i\right)\, \simeq \underset{i \in I}{\colim}\,  \mathrm{oblv}(B_i)\,.
\end{equation}
We now fix an algebra $A$, and because tensoring with the vector space $\mathrm{oblv}(A)$ preserves all colimits, we have that 
\begin{align}
     \mathrm{oblv}\left( A \otimes \underset{i\in I}{\colim}\, B_i\right) &\simeq  \mathrm{oblv}(A)\otimes \underset{i \in I}{\colim}\,  \mathrm{oblv}(B_i)\\ \notag
     &\simeq  \underset{i \in I}{\colim}\left(  \mathrm{oblv}(A)\otimes  \mathrm{oblv}(B_i) \right)\\
     &\simeq \mathrm{oblv}\left(\underset{i \in I}{\colim}\left( A \otimes B_i  \right)\right)\,.
\end{align}
Thus, we have $A \otimes \colim_{i\in I}B_i\simeq \colim_{i\in I} A\otimes B_i$, and by symmetry, tensoring preserves sifted colimits in either variable separately.
    
    We invoke the results of \cref{lem:sectorcompose}  and \cref{prop:composenets} which show that the defect $n$-nets of codimension-$k$, for $1\leq k \leq n$,  close under composition in order to treat them as a simplicial space of $k$-morphisms. Furthermore, the collar excision condition for defect $n$-nets ensures that their gluing is compatible with the decomposition of a morphism into a sequence of composable morphisms. This equips topological $n$-nets with the structure of an algebra having $n$ independent directions of associative unital multiplications. We encode this structure in the simplicial space $\Alg_n(\cC^\square_{\vec{\bullet}})_{\vec{\bullet}}$, where the outer label $\vec{\bullet}$ records the codimension of the defect $n$-net and, correspondingly, parametrizes its spaces of morphisms.
    Finally, we use \cref{thm:buildingAlg} to show that  $\Alg_n(\cC^\square_{\vec{\bullet}})_{\vec{\bullet}}$ satisfies the Segal condition, is complete, and inherits a symmetric monoidal structure, and thus by definition is a symmetric monoidal $(\infty,n+1)$-category $\Alg_n(\cC)$. By using \cref{def:Algn} for $\Alg_n(\cC)$, we can translate every morphism to an equivalent description using nets. This gives the description of $\Net_n$ that is presented in the theorem.
\end{proof}

\section{Higher Karoubi Completion}\label{section:condensation}

The Karoubi completion, otherwise known as idempotent completion, of a category is an enlargement of the category to one where all idempotents split.
A higher categorical generalization of Karoubi completion goes by the name \textit{condensation completion} \cite{Gaiotto:2019xmp}. 
The use of condensation completion for establishing TQFTs comes from the insight that categories of dualizable objects often arise as the idempotent completion of a simpler category.  A simple example is that over a ring $R$, dualizable $R$-modules are precisely the finitely generated projective $R$-modules. These are retracts of finitely free $R$-modules. The category $\tVect^{\fd}_\mathbb{C}$, of fully dualizable 2-vector spaces over $\mathbb{C}$ is a perfectly good target for (1+1)d TQFTs. It serves as the Morita category of finite semisimple separable unital algebras over $\C$, and is constructed as follows. Start with the 1-category of finite-dimensional vector spaces over $\C$, $\Vect^{\fd}_\C$, together with its symmetric monoidal tensor product, deloop it to a one-object 2-category and perform the condensation completion, i.e. complete under a lax notion of idempotents. Going higher for all $n$, by repeating the process of delooping and taking condensations, one recovers the fully dualizable $n$-categories of $n$-vector spaces, denoted $\mathbf{nVect}^{\fd}_\mathbb{C}$.

We review condensation completion in this section, following the work of \cite{Gaiotto:2019xmp,Kong:2024ykr}, with an eye toward constructing the category of quantum spin systems used in \cref{thm:QS}, much as it is used to construct categories that serve as natural targets for TQFTs. We review the concept of an $n$-condensation which categorifies a split surjection and an $n$-condensation monad, which categorifies idempotents. We also discuss the concept of condensation modules and bimodules, leading to a main result given in \cref{prop:unique} about the uniqueness of condensations. Another crucial concept that we discuss is how unitality is defined for $n$-condensation monads, and the relevance of unitality for defining the fully dualizable categories that are relevant for our study of quantum spin systems.


We begin with a discussion of idempotent completion, before categorifying the concept. Let $X$ and $Y$ be objects of a 1-category $\cC$. Let $f:X\to Y$ and $g:Y \to X$ be a pair of 1-morphisms such that $f\circ g = \id_Y$ and $e= g \circ f: X\to X$ is an idempotent. Then if $Y$ is a direct summand of $X$, it is realized as the image of $e$. Not every idempotent in a general category admits an image. The following definition establishes a setting in which every idempotent does.
\begin{defn}
    Given a 1-category $\cC$, let $\widehat{\cC}$ be the \textit{Karoubi completion of $\cC$}. The objects of $\widehat{\cC}$ are pairs $(X,e)$ where $X$ is an object of $\cC$, and $e:X\to X$ is an idempotent. The morphisms consist of triples $(e,f,e'):(X,e)\to (X',e')$ where $f:X \to X'$ and $f \circ e = e' \circ f = f$.
\end{defn}

Moreover, $\widehat{\cC}$ is the universal category which contains $\cC$ and for which all idempotents have an image.

\begin{defn}
    Let $X$ and $Y$ be objects of a weak $n$-category. A \textit{$n$-condensation} of $X$ onto $Y$ is a pair of 1-morphisms $f:X\to Y$ and $g:Y\to X$ together with a $(n-1)$-condensation $fg$ to $
    \id_Y$.
\end{defn}

We will use the condensed notation $X \condense Y$, to denote the $n$-condensation of $X$ onto $Y$.

\begin{example}
   In  a 1-condensation, the condensation $fg \condense \id_Y$ is just $fg = \id_Y$, which forces $f: X \twoheadrightarrow Y$ to be a epimorphism and $g:Y \hookrightarrow X$ an monomorphism. In this case, $Y$ is said to be a retract of $X$.
\end{example}

\begin{defn}\label{def:walkingcondensation}
    The \textit{walking $n$-condensation} is the  2-object $n$-category $\spadesuit_n$  generated by two objects $X$ and $Y$, with two 1-morphisms $f:X\to Y$ and $g:Y \to X$, two 2-morphisms $\phi :fg \to \id_Y$ and $\gamma:\id_Y \to fg$, two 3-morphisms, etc, ending with two $n$-morphisms $\alpha$ and $\beta$ and one equation $\alpha \beta = \id$.
\end{defn}
\noindent A succinct way to think of an $n$-condensation in a category $\cC$ is as a diagram in $\cC$ of shape $\spadesuit_n$, i.e. a functor $\spadesuit_n\to \cC$. We represent the data of the $n$-condensation as $(X,Y,f,g,\ldots)$, and diagrammatically as

  \begin{equation} \label{eqn.cond}
\begin{tikzpicture}[baseline=(midpoint)]
  \path (0,0) node (X) {$X$} (0,-2) node (Y) {$Y$} (0,-1) coordinate (midpoint);
  \draw[condensationlarge] (X) -- (Y);
\end{tikzpicture}
\quad = \quad
\begin{tikzpicture}[baseline=(midpoint)]
  \path (0,0) node (X) {$X$} (0,-.35) coordinate (Xe) (0,-2) node (Y) {$Y$} (0,-1) coordinate (midpoint);
  \draw[->] (X) .. controls +(-.75,-.75) and +(-.75,.75) ..  node[auto,swap] {$\scriptstyle f$} (Y);
  \draw[->] (Y) .. controls +(.75,.75) and +(.75,-.75) .. node[auto,swap] {$\scriptstyle g$} (X);
  \path (Y.north) +(.1,0) coordinate (Ytopa) +(-.1,0) coordinate (Ytopb);
  \draw[->] (Ytopa) .. controls +(.6,.6) and +(.5,0) .. (Xe) .. controls +(-.5,0) and +(-.6,.6) .. (Ytopb);
  \draw[doublecondensationlarge] (Xe) -- (Y);
\end{tikzpicture}\,. \end{equation}

\begin{rem}
    For each $\spadesuit_n$, we can define a commuting $n$-condensation square via the diagram $\spadesuit^2_n$, which is a product of $n$-categories. The product of $n$-categories is computed by computing the product on $i$-morphisms for every $i$. In the case of $\spadesuit^2_n$, one has four objects that label the vertices of a square, and the remaining morphism data in the condensation label the four sides of the square:
$$
\spadesuit^2_2 \qquad = \qquad 
\begin{tikzpicture}[scale=2, baseline=(Y1.base)]
  \path (0,0) node (X) {$X$}
     (-1,-1) node (Y1) {$Y_1$}
     (1,-1) node (Y2) {$Y_2$}
     (0,-2) node (Z) {$Z$};
  \draw[->] (X) .. controls +(-.5,-.25) and +(.25,.5) .. node[auto,swap] {$\scriptstyle f_1$} (Y1);
  \draw[->] (Y1) .. controls +(.5,.25) and +(-.25,-.5) .. node[auto,swap] (g1) {$\scriptstyle g_1$} (X);
  \draw[->] (X) .. controls +(.25,-.5) and +(-.5,.25) .. node[auto,swap] (f2) {$\scriptstyle f_2$} (Y2);
  \draw[->] (Y2) .. controls +(-.25,.5) and +(.5,-.25) .. node[auto,swap] {$\scriptstyle g_2$} (X);
  \draw[->] (Y1) .. controls +(.25,-.5) and +(-.5,.25) .. node[auto,swap] {$\scriptstyle p_1$} (Z);
  \draw[->] (Z) .. controls +(-.25,.5) and +(.5,-.25) .. node[auto,swap] (q1) {$\scriptstyle q_1$} (Y1);
  \draw[->] (Y2) .. controls +(-.5,-.25) and +(.25,.5) .. node[auto,swap] (p2) {$\scriptstyle p_2$} (Z);
  \draw[->] (Z) .. controls +(.5,.25) and +(-.25,-.5) .. node[auto,swap] {$\scriptstyle q_2$} (Y2);
  \draw[doublecondensationlarge] (X) -- (Y1); \draw[doublecondensationlarge] (X) -- (Y2); \draw[doublecondensationlarge] (Y1) -- (Z); \draw[doublecondensationlarge] (Y2) -- (Z); 
   \path (0,-1) node {$[\cdots]$};
\end{tikzpicture} 
$$

Call the object $X$ the ``top'' object and $Z$ the ``bottom'' object. The paths from the top object to the bottom objects commute, which in higher categories, corresponds to the fact that there are coherent isomorphisms between different compositions from the top to bottom objects.
The use of $\left[\ldots\right]$ in the interior of the square is to denote the products of higher cells that are also generated, and the commutativity isomorphisms involving these cells.    
    Going further, one can describe a commuting $n$-condensation $k$-cube as a diagram of shape $\spadesuit^k_n$.
\end{rem}

To categorify the notion of an idempotent in an $n$-category, we view  $e:X\to X$ as an object of $\End_{\spadesuit_n}(X)$. The equation $e^2=e$ in 1-categories categorifies to a condensation $e^2 \condense e$. This arises from the relations 
\begin{equation}
    e^2 = (gf) (gf) = g (fg) f \condense g \,\id_Y f = gf = e.\footnote{We have used the property of 2-morphisms in \cref{def:walkingcondensation}, with the notion that they do not compose to $\id_Y$, but condense to $\id_Y$.}
\end{equation}
In this equation, we can use the relation $g \phi f$, but there is also a backwards map $ g \gamma f$. These can be thought of as a binary operation and binary co-operation on $e$, making $e$ into a condensation (bi)algebra. The next definition gives the categorification of an idempotent.

\begin{defn}
    A \textit{condensation monad} in an $n$-category is an object $X\in \cC$ and a condensation algebra object in the monoidal $(n-1)$-category $\End_\cC(X)$. We will denote the condensation monad on $X$ in terms of the data $(X,e,\ldots)$.
\end{defn}

\begin{defn}\label{def:condensate}
    The \textit{condensate} of a condensation monad  $(X,e,\ldots)$ in $\cC$, if it exists, is an extension of the monad into a condensation $X \condense Y$ in $\cC$. A $n$-category \textit{has all condensates} if all hom $(n-1)$-categories have all condensates and every condensation monad extends to a condensation. 
\end{defn}

Given two condensation monads $(X_1,e,\ldots)$ and $(X_2,e_2,\ldots)$ in $\cC$ we can consider a 1-morphism $m: X_1 \to X_2$ which induces a condensation bimodule between the two condensation monads. 

\begin{defn}
    Let $(X_1,e_1,\ldots)$ and $(X_2,e_2,\ldots)$ be two $n$-condensation monads in an $n$-category $\cC$. A \textit{condensation bimodule}  between the $n$-condensation monads consists of a 1-morphism $m: X_2 \to X_1$, and for each pair of natural numbers $(i,j)$, an $(i+j)$-dimensional commuting condensation cube whose top object is $e^i_1 m e^j_2$ and bottom object is $m$. The edges of the cube are various condensations of the form $e^k_1 m e^l_2 \condense e^{k-1} m e^{l}_2$, \,\,$e^k_1 m e^l_2 \condense e^{k} m e^{l-1}_2$.. for natural numbers $k\leq i$ and $l \leq j$. The various faces are whiskerings of lower dimensional cubes or identities.\footnote{One can think of the combinatorics that dictate which faces are whiskerings and which are identities as being the same combinatorics that axiomatize associative bimodules in terms of commuting cubes.} 
\end{defn}

\begin{defn}
    A \textit{left module} for a condensation monad $(X_1,e,\ldots)$ is a bimodule between $(X_1,e,\ldots)$ and $(X_2,\id_{X_2},\ldots)$, and  a \textit{right module} is a bimodule between $(X_2,\id_{X_2},\ldots)$ and $(X_1,e,\ldots)$, for some object $X_2$.
\end{defn}

\begin{lem}\label{lem:modules}
    If $(X,Y,f,g,\ldots)$ is a condensation, then $f$ and $g$ are right and left modules for $e = gf$, respectively.
\end{lem}

\begin{proof}
    We show that $g$ is a left module for $e$, and the same strategy can be used to show that $f$ is a right module for $e$. Suppose we have a three object category with objects labeled by $X$, $Y$ and $X_2$, and $X_2$ is only equipped with the identity morphism to itself.
    Let $\mu: X_2 \to Y$, so that $m = g \mu: X_2 \to X$. The map $\mu$ gives an isomorphism $\hom(Y,X)\cong \hom(X_2,X)$. The condensation monad $e$ is a $e\text{-}e$-bimodule, and so we get a sequence of condensation cubes $e^i m \condense m$.
\end{proof}

If the condensation monads are in an $n$-category, then the condensation bimodules form an $(n-1)$-category whose higher morphisms are condensation bimodule homomorphisms. Moreover, this $(n-1)$-category admits a natural tensor product of condensation bimodules. To define this tensor product, it will be simpler to work with condensation modules. 

Let $(X,e,\ldots)$ be a condensation monad, and let $X'_1 \xleftarrow{m_1} X$ and $X \xleftarrow{m_2} X'_2$ be left and right $e$-modules respectively. Suppose that  $X\condense Y$ and $e = gf$, so that $m_1 = \mu_1 f$ and $m_2 = g \mu_2$ for morphisms $\mu_1: Y \to X'_1$ and $\mu_2: X'_2 \to Y$. This is depicted in the following diagram: 

\begin{equation} \label{eqn.tensor1}
\begin{tikzpicture}[baseline=(midpoint)]
  \path (0,0) node (X) {$X$} +(0,-.35) coordinate (Xe) +(0,-2) node (Y) {$Y$} +(0,-1) coordinate (midpoint);
  \draw[->] (X) .. controls +(-.75,-.75) and +(-.75,.75) ..  node[auto] {$\scriptstyle f$} (Y);
  \draw[->] (Y) .. controls +(.75,.75) and +(.75,-.75) .. node[auto] {$\scriptstyle g$} (X);
  \draw[doublecondensation] (Xe) -- (Y);
  \path (3,-2) node  (X2) {$X_2'$};
  \path (-3,-2) node  (X1) {$X_1'$};
  \draw[->] (X2) -- node[auto] {$\scriptstyle \mu_2$} (Y);
  \draw[<-] (X1) -- node[auto,swap] {$\scriptstyle \mu_1$} (Y);
  \path 
    (X2) +(-.5,.25) coordinate (X2m)
    (X1) +(.5,.25) coordinate (X1m)
    (X) +(.5,-.25) coordinate (Xm2)
    (Y) +(.5,.25) coordinate (Ym2)
    (X) +(-.5,-.25) coordinate (Xm1)
    (Y) +(-.5,.25) coordinate (Ym1)
    ;
  \draw[->]  (X2m -| X2.west) -- node[auto,swap] {$\scriptstyle m_2$} (Ym2) .. controls +(.5,.5) and +(.5,-.5) .. (Xm2);
  \draw[<-]  (X1m -| X1.east) -- node[auto] {$\scriptstyle m_1$} (Ym1) .. controls +(-.5,.5) and +(-.5,-.5) .. (Xm1);
\end{tikzpicture}
\end{equation}

From the maps in \cref{eqn.tensor1}, one observes a condensation $m_1m_2 \condense \mu_1 \mu_2$, and hence we take the condensation algebra corresponding to this to be $\epsilon \in \End(m_1 m_2)$.

\begin{defn}\label{def:moduletensor}
    Let $\cC$ be an $n$-category for which all of its hom $(n-1)$-categories have all condensates. Let $(X,e,\ldots)$ be a condensation monad, and let $X'_1 \xleftarrow{m_1} X$ and $X \xleftarrow{m_2} X'_2$ be left and right $e$-modules respectively. Suppose that  $X\condense Y$ and $e = gf$, so that $m_1 = \mu_1 f$ and $m_2 = g \mu_2$ for morphisms $\mu_1: Y \to X'_1$ and $\mu_2: X'_2 \to Y$. 
    The \textit{tensor product} $m_1 \otimes_e m_2: X'_2 \to X'_1$ is the condensate corresponding to the  condensation monad  $(m_1 m_2, \epsilon,\ldots)$.
\end{defn}

The fact that the condensation monad $(m_1 m_2, \epsilon,\ldots)$ has a condensate is due to \cite[Proposition 2.3.9]{Gaiotto:2019xmp}. In particular the condensate is $\mu_1 \mu_2$ and hence $m_1 \otimes m_2 \cong \mu_1 \mu_2$. One may wonder why the tensor product, defined as a condensate in \cref{def:moduletensor}, is well defined.  
Using the background we have provided we can show the following uniqueness result for condensates.
\begin{thm}[{\cite[Theorem 2.3.2]{Gaiotto:2019xmp}}]\label{prop:unique}
    Let $(X,e,\ldots)$ be a condensation monad in $\cC$ an $n$-category where all hom $(n-1)$-categories have all condensates. Then the category of condensates for the  condensation monad  $(X,e,\ldots)$  is either empty or contractible.  
\end{thm}
\begin{proof}
    We first show that each condensation monad has at most one condensate. Suppose $(X,e,\ldots)$ has two condensates $(X,Y_1,f_1,g_1,\ldots)$ and $(X,Y_2,f_2,g_2,\ldots)$, then it would suffice to show that the two are isomorphic. By \cref{lem:modules}, $f_i$ and $g_i$ are modules for $e$, and hence we can form the following morphisms  $f_1 \otimes_e g_2: Y_2 \to Y_1$ and $f_2 \otimes_e g_1: Y_1 \to Y_2$ by tensoring modules. We obtain the identity map $Y_1 \to Y_1$ by looking at the composition
    \begin{equation}
        f_1 \otimes_e g_2 f_2 \otimes g_1 = f_1 \otimes_e e \otimes_e g_1 \cong  f_1 \otimes_e g_1 \cong \id_{Y_1}. 
    \end{equation}
    Similarly, we also obtain the identity map $Y_2 \to Y_2$ via the composition
   \begin{equation}
        f_2 \otimes_e g_1 f_1 \otimes g_2 = f_2 \otimes_e e \otimes_e g_2 \cong  f_2 \otimes_e g_2 \cong \id_{Y_2}. 
    \end{equation}
Therefore, the two condensations are isomorphic. 
To prove the second part of the theorem, suppose $(X,Y,f,g,\ldots)$ is a condensate of  $(X,e,\ldots)$. We want to show that the category of those endomorphisms of $(X,Y,f,g,\ldots)$ that restrict to the  identity endomorphism of $(X,e,\ldots)$ is contractible. Consider a nontrivial endomorphism of $(X,Y,f,g,\ldots)$ that restricts to  the identity morphism on $(X,e,\ldots)$.
A condensation monad such as $(X,e,\ldots)$ with the identity endomorphism can only condense to $(X,Y,f,g,\ldots)$, where $Y$ also has the trivial endomorphism. If  $(X,Y,f,g,\ldots)$ is equipped with a nontrivial endomorphism, then using the fact that condensation monads have at most one  condensate,  any nontrivial endomorphism on $(X,Y,f,g,\ldots)$  cannot restrict to the identity morphism on $(X,e,\ldots)$. Thus the category of those endomorphisms of $(X,Y,f,g,\ldots)$ that restrict to the  identity endomorphism of $(X,e,\ldots)$ is contractible. The fact that there are no higher morphisms of condensations follows from the same logic, where now we think of $(X,Y,f,g,\ldots)$ equipped with a higher endomorphism and show that it cannot restrict to  $(X,e,\ldots)$ with the identity higher endomorphism.
\end{proof}

\begin{thm}[{\cite[Theorem 2.3.11]{Gaiotto:2019xmp}}]\label{thm:condensation}
    Suppose that $\cC$ is an $n$-category all of whose hom $(n-1)$-categories are condensation complete.
We will reuse notation, and denote by $\widehat{\cC}$ the condensation completion of $\cC$. 
\begin{itemize}
    \item The objects of $\widehat{\cC}$ are condensation monads of $\cC$, 1-morphisms are condensation bimodules and composition is given by tensor product of bimodules. The higher morphisms are homomorphisms of condensation bimodules.
    \item $\widehat{\cC}$ has condensates for all condensation monads.
    \item The functor $\cC \to \widehat{\cC}$ sending $X \mapsto (X,\id_X,\ldots)$ is fully faithful, and an equivalence if $\cC$ has all condensates.
    \item The construction $\cC \mapsto \widehat{\cC}$ is functorial in $\cC$, and takes products of $n$-categories to products.
    \item If $\cD$ is any $n$-category with all condensates, then any functor $\cC \to \cD$ extends uniquely to a functor $\widehat{\cC}\to \cD$.
\end{itemize}
\end{thm}

The major application of \cref{thm:condensation} is to build the higher category $\npVect^{\fd}_\mathbb{C}$ via the following procedure. Given a monoidal $(n-1)$-category $\cC$ with all condensates, take $\rB \cC$ to be the one-object monoidal $n$-category. Applying \cref{thm:condensation} we denote by $\widehat{\rB \cC}$, the universal $n$-category with all condensates and a fully faithful map $\rB \cC\to \widehat{\rB \cC}$. If $\cC$ is in fact symmetric monoidal, then one can iteratively deloop by taking $\rB$ of $\cC$ and then condensation completing. 
\begin{example}
    Delooping and condensation applied to $\cC = \mathbb C$ gives $\widehat{\rB \cC}= \Vect^{\fd}_\C$. Further applying it
$n$ times gives the $(n+1)$-category 
\begin{equation}\label{eq:n1vect}
    \widehat{\phantom{aa}\rB \ldots\widehat{\rB \cC}} = \npVect^{\fd}_\C\,.
\end{equation}
\end{example}

So far it is not clear why full dualizability should be preserved by iteratively delooping and condensation completing. We will address this in the next section.
\subsection{Unital Condensation and Full Dualizability}

Thus far, the discussion of condensation algebras has not required a unitality condition. Consequently, even at the level of \!\!\!\!\!$\widehat{\phantom{aa}\rB\Vect^{\fd}_\C}$, the objects comprising this category should be regarded as a class of nonunital algebras. At this stage, however, this presents an obstruction to interpreting the category as a target for TQFTs, for which one requires unital associative algebras. We thus need an upgrade to unital condensation.
\begin{defn}
    A $n$-condensation $(X,Y,f,g,\ldots)$ is \textit{0-unital} if $\phi: fg \to \id_Y$ is the counit of the adjunction $f \dashv g$.\footnote{ By $\phi$ being a counit, this implies there is a unit $\eta: \id_X \to gf = e$ such that the compositions
\begin{align}\notag 
    f \xrightarrow{f\eta} fgf \xrightarrow{\phi f} f\,, \\
    g \xrightarrow{\eta g} gfg\xrightarrow{g \phi} g
\end{align}
are equivalent to the identity. In fact the unit $\eta$ is exactly what is necessary to make the endomorphism $e = gf: X\to X$, which is an associative algebra object in $\End(X)$, into a unital algebra.} For $n\geq 3$, a $n$-condensation is \textit{1-unital} if the
corresponding $(n- 1)$-condensation $(f\circ g,\id_Y,\ldots)$ is 0-unital. For any integer
$0 \leq k \leq n -1 $, we define inductively the notion of a \textit{$k$-unital} $n$-condensation, if the corresponding $(n-1)$-condensation is $(k-1)$-unital.
\end{defn}

\noindent If a $n$-condensation is $k$-unital for all $k$ then it is called fully unital. 
It was argued in \cite[Theorem 3.1.7]{Gaiotto:2019xmp} that for a $n$-category $\cC$ that has adjoints for 1-morphisms and all of whose hom $(n-1)$-categories have all condensates, then every $n$-condensation monad $(X,e,\ldots)$ in $\cC$ is Morita equivalent as an object in $\widehat{\cC}$ to an object represented by a 0-unital $n$-condensation monad. Furthermore, if $\cC$ has adjoints for  $k$-morphisms where $1\leq k \leq n-1$ then every $n$-condensation monad is Morita equivalent to a $(k-1)$-unital one in $\widehat{\cC}$. If $\cC$ has adjoints for morphisms at all levels, one would naturally expect that every $n$-condensation monad is Morita equivalent to a fully unital one. However, this is not a consequence of \cite[Theorem 3.1.7]{Gaiotto:2019xmp}. To proceed with the subsequent results and formulate our analysis in arbitrary dimensions, we make the following assumption. We expect that the details underlying this assumption can be established in future work.

\begin{assume}\label{assume:extendunital}
    Let $\cC$ be an $n$-category that has adjoints for morphisms of all levels. Then every $n$-condensation monad is Morita equivalent to a fully unital one.
\end{assume}

\begin{rem}\label{rem:adjoints}
    An application of \cref{assume:extendunital} that will prove useful later is that if all morphisms have adjoints then in the condensation  $(X,Y,f,g,\phi,\gamma,\ldots)$, arising from the splitting of the unital condensation monad $(X,e,\ldots)$, the map $f$ has a right adjoint $f^R$. Furthermore, the condensation lifts to a unital condensation $X \condense Y$, whose 1-morphisms are $f$ and $f^R$. Similarly, if $g$ admits a left adjoint $g^L$ then there is a unital condensation $X \condense Y$ whose 1-morphisms are $g$ and $g^L$.
\end{rem}

\begin{rem}
     Given a fully unital $n$-condensation monad $(X,e,\ldots)$, one can consider the fully unital $n$-condensation algebra $(e,\ldots)$ as a \textit{separable algebra} in the monoidal $(n-1)$-category $\End(X)$. 
    Consider the case when the category level $n=2$. A separable algebra in this case is a unital separable Frobenius algebra. In the case $n = 3$, this notion of a separable algebra
coincides with the notion of a rigid separable algebra given in \cite{D7}.
\end{rem}

Furthermore we have the following equivalence relation on condensation monads.
\begin{lem}
    If an associative monad $e$  in $\cC$ extends to a  condensation monad structure, then all extensions are canonically equivalent as objects of $\cC$\,.
\end{lem}

We now move onto discussing dualizability in higher categories, generalizing the concept of rigid in 1-categories, which means that all objects admit left and right duals. In an $(n+1)$-category there are different levels of duals, and we say that an $(n+1)$-category is $m$-dualizable if all $k$-morphisms $1\leq k \leq m$  admit both left and right adjoints, and a monoidal $n$-category $\cC$ is $n$-dualizable if in the $(n+1)$-category $\rB \cC$ all morphisms $0\leq k \leq m$ admit both left and right adjoints. We say that a monoidal $n$-category is fully dualizable if it is $n$-dualizable.

We conclude our review with the final theorem explaining why the objects of a condensation-complete category are fully dualizable, and hence provide natural targets for the description of condensed matter systems and TQFTs.

\begin{thm}[{\cite[Theorem 4.1.2]{Gaiotto:2019xmp}}]
\label{thm:fulldual}
   If $\cC$ is a fully dualizable symmetric monoidal $n$-category with all condensates, then the symmetric monoidal $(n+1)$-category $\widehat{\rB \cC}$ is also fully dualizable. 
\end{thm}

Applying this theorem to $\mathbf{(n+1)Vect}^\fd_{\C}$, in \cref{eq:n1vect}, shows that it is fully dualizable.

\section{The Category of Quantum Spin Systems}\label{section:catofQS}
After establishing the category of topological $n$-nets in the previous section, we will now define the category of $n$-dimensional quantum spin systems $\QS_n$ as a condensation complete subcategory of $\Net_n$. This perspective of identifying suitable subcategories that model physically relevant systems is naturally motivated by the functorial formulation of TQFTs. In particular, the cobordism hypothesis tells us that the space of TQFTs is not obtained by taking all objects of a target symmetric monoidal category (in our setting, $\Net_n$), but only its fully dualizable objects (in our setting, $\QS_n$).

\begin{rem}
    In \cite{Kong:2022hjj} the authors have a related category called the category of ``local quantum symmetries'' that is defined with respect to von Neumann algebras. This category serves to capture the symmetry aspects of topological phases that do not depend on extended topological operators. The authors also introduce a category of ``quantum liquids'', intended to model lattice topological phases. Conceptually, this is closely aligned with the class of systems we seek to capture.
\end{rem}

In this section we show that $\QS_n$ has the properties of a Morita category of higher fusion categories, given by  $\npVect^{\fd}_{\mathbb{C}}$. This gives a  translation between condensations of categories, and condensations in terms of nets.  The construction of $\QS_n$ is done iteratively and utilizes $\QS_{n-1}$ as its category of defect $n$-nets of codimension greater than 0, i.e.  if $\cA$ is the trivial topological $n$-net then $\End_{\QS_n}(\cA) \simeq \QS_{n-1}$.

To finish the  construction one must show that for every condensation monad on a topological $n$-net $\cA \in \QS_n$ which admits a condensate in the sense of \cref{def:condensate}, that there is a way to build this condensate. 
 The starting point for constructing $\QS_n$ will be to first consider the subcategory $\Net^{\fd}_n \subset \Net_n$, consisting of the \textit{finite semisimple} topological $n$-nets and defect $n$-nets, see \cref{def:semisimpleNet} and  \cref{def:finiteNet}. We will then condensation complete $\Net^{\fd}_n$, i.e. appending all condensates for condensation monads, to obtain $ \widehat{\Net^{\fd}_n}= \QS_n$.\footnote{Our construction actually should be done on the trivial topological $n$-net, thereby exhibiting every $n$-net as a condensate of the monoidal unit. Nevertheless, we formulate it using a condensation monad on an arbitrary net $\mathcal A$. This entails no loss of generality, since a change of condensation monad can be implemented by a condensation bimodule. In particular, in any of the arguments in \S\ref{subsection:2net} and \S\ref{subsection:nnet} that constructions an algebra which is assigned by a net, can be carried out equally well using the trivial net, with the relevant algebras realized as endomorphism algebras of vector spaces over $\cC$.}
 
\begin{rem}
    Even though the theory of unital condensations is not quite complete enough to make sure that $\npVect^{\fd}_{\mathbb{C}}$ is exactly the Morita category of higher fusion categories for all $n$, hence the need for \cref{assume:extendunital}, we will really only need the Karoubi/condensation completion property of $\npVect^{\fd}_{\mathbb{C}}$ for our analysis of $\QS_n$. We will not be too concerned with how $\npVect^{\fd}_{\mathbb{C}}$ is formally constructed as a fusion category, namely from a form of iterative delooping and condensation completion, but just that it exists. Future work will describe how to formally construct it coherently and fix the issues in \cite{Gaiotto:2019xmp} regarding unitality.
\end{rem}

\begin{rem}
   Quantum spin systems are also commonly formulated in the operator-algebraic setting using C$^*$-algebras or von Neumann algebras, where unitarity plays a central role. If topological $n$-nets are instead defined to assign von Neumann algebras to disk-like regions, the resulting category of quantum spin systems should exhibit the structure of a higher Hilbert space 
    \cite{Kong:2022hjj}.
   In a similar vein, the $n$-disk category also bears a relationship to higher Hilbert spaces as shown in \cite{Wesley2026}. Unitary constructions for higher Hilbert space categories have also been developed in \cite{chen2026manifestly,ferrer2026many}.
\end{rem}

Given the discussion in \S\ref{subsection:justification}, our framework allows the classification of QCAs to be done as follows: 
\begin{enumerate}
    \item We first reformulate quantum spin systems in terms of topological $n$-nets, in which the 1-morphisms capture the same equivalence conditions used for classifying QCA up to stable equivalence and quantum circuits.
    \item The first step gives a Morita category description of quantum spin systems. Using our knowledge of the Morita category of condensation complete categories we then compute the homotopy groups of the core of $\QS_4$ which allows us to access QCAs up to (3+1)d. This will establish \cref{mainthm:QCA}.
\end{enumerate}

We show the following result in order to apply \cref{assume:extendunital} to all the condensation monads that follow. The strategy for the proof is essentially the same as \cite[Theorem 2.20]{BDH5}, which was done for the 3-category case with conformal nets.

\begin{prop}\label{prop:alladjoint}
    The category $\Net^\fd_n$ has adjoints for all $k$-morphisms where $0\leq k \leq n$.
\end{prop}
\begin{proof}[Proof sketch]
    At the level of objects, the dual for $(\cA,H_\cA)$ is given by $(\cA^{\op},\overline{H}_\cA)$ where $\cA$ is the opposite algebra with reversed orientation for composition, and $\overline{H}_\cA$ is the dual vector space. The dual for a $1\leq k <n$ morphism in the form of a defect $n$-net of codimension-$k$ $(\cF,H_\cF)$ is $(r^*_k\cF, \overline{H}_\cF)$ where $r_k$ is the reflection across the hyperplane $x_{n-k}=\tfrac{1}{2}$ of the space $M = (0,1)^n$ where all defect $n$-nets are defined. In particular, if $\cF$ is thought of as a morphism of defect $n$-nets of codimension-$(k-1)$, which starts at $\mathcal{M}$ and ends at $\mathcal{N}$, then $r^*_k \cF$ is a morphism from $r^*_{k-1}\mathcal{N}$ to $r^*_{k-1}\mathcal{M}$. 
    The full proof that compositions of dualizable $k$-morphisms are dualizable uses the same sewing arguments as \cite[Section 2.C]{BDH5}. 
\end{proof}

\subsection{Condensation for $1$-nets}

We warm up by showing that every condensation monad on a finite  topological $1$-net with semisimple sector $\cA \in \QS_1$ admits a condensate which we can construct. In particular, a topological 1-net $(\cA,H_\cA)$ is finite semisimple if $H_\cA$ is finite and semisimple as the regular module for $\cA$. This is the first level to consider because a 1-net is the lowest possible dimension of net that we can work with, corresponding to (0+1)d quantum spin systems.

\begin{prop}
    The category $\QS_1$ can be identified with the symmetric monoidal 2-category of finite semisimple algebras, finite semisimple bimodules, and bimodule maps, i.e. $\mathbf{2Vect}^{\fd}_\C$.
\end{prop}

\begin{proof}
   The data of a defect 1-net of codimension-1 $(\cF,\phi_\cF)$  (i.e. a sector) between  two topological 1-nets $(\cA,H_\cA)$ and $(\cB,H_\cB)$ is equivalent to a bimodule for a pair of associative algebras. $\cF$ is a morphism in $\QS_1$ if $\cA$ and $\cB$ are finite semisimple as 1-nets. Hence at the level of objects of $\QS_1$, $\cF$ is a finite semisimple $\cA\text{-}\cB$-bimodule, and $\phi_\cF: H_\cA \to H_{\cB} \otimes_{\cB} \cF$ is a map of $\cA$-modules.
A 2-morphism in $\QS_1$, i.e. a homomorphism of sectors, is a bimodule map $\xi : (\cF,\phi_\cF) \to (\mathcal{G},\phi_\mathcal{G})$ between finite semisimple bimodules, such that $\xi : \cF \to \cG$ and $\xi \circ \phi_\cF  = \phi_\cG$. 

The description of the collection of 1 and 2-morphisms identifies them with the objects and morphisms of the category $\Vect^\fd_\C$, which is condensation complete. Therefore, all we need to show is how to construct condensates for condensation monads, at the level of objects.  Suppose that $e$ is a condensation algebra on a topological 1-net $\cA$, where $e$ is defined as an object in $\widehat{\Net^{\fd}_0} = \Vect^{\fd}_\C$. Suppose we extend the condensation monad $(\cA,e,\ldots)$ to a condensation $\cA \condense \cB$. We hence need to show that indeed $\cB$ is a finite semisimple topological 1-net, which would imply that $\cB$ defines an element of $\tVect^{\fd}_\C$.

From the condensation  $\cA \condense \cB$ we will write the condensation on $\cA$ in terms of the data $(\cA,\cB,f,g,\ldots)$. Using \cref{assume:extendunital} and \cref{prop:alladjoint}, we can upgrade $e$ to a unital condensation. In this case $f$ has a right adjoint $f^R$, by \cref{rem:adjoints}. The condensation algebra $e$ is semisimple as an $\cA$-sector, and the left action of $\cA$ on $e$ factors through $\cB$. Hence the map $f^R: \cB \to \cA$ is semisimple and we can choose to regard $f^R  = H_\cB$ as a sector. As a consequence $f: \cA \to \cB$ must also be semisimple as a sector. Consider now the composition $f \circ f^R$, for which $\id_\cB$ is a direct summand. Because $f$ is a sector for a finite net $\cA$ and $f^R$ is dual to $f$ as a sector, the composite  $f \circ f^R$ is still finite as a sector, and thus so is $\id_\cB$ since finiteness is a property preserved by direct summands. This establishes that $\cB$ is finite.
\end{proof}

The full dualizability of $\QS_1$ will be guaranteed by \cref{thm:fulldual}. Indeed after one iteration of delooping and condensation completing starting with $\Vect^\fd_\C$ we obtain \!\!\!\!\!$\widehat{\phantom{aa}\rB\Vect^{\fd}_\C} = \mathbf{2Vect}^\fd_\C$.

\subsection{Condensation for Topological $2$-nets}\label{subsection:2net}

We now show that every condensation monad on a topological $2$-net $\cA$ admits a condensate, which we can construct. Furthermore, every condensation bimodule is induced by a 1-morphism in $\QS_2$, and bimodule maps are induced by 2-morphisms, etc. The strategy for constructing the $2$-net condensate involves first constructing its sector, and recovering the algebra that the net assigns by looking at endomorphisms of the sector. 
We present the generalization of this construction in higher dimensions in the next section. 

\begin{prop}\label{prop:QS2}
    The category $\QS_2$ is equivalent to the symmetric monoidal 3-category  $\mathbf{3Vect}^{\fd}_\C$.
\end{prop}

\begin{proof}

Let $e$ be a condensation algebra on $\cA$, with unit $\eta: \id_\cA \to e$. In the condensation $e^2 \condense e$ let $m: e^2 \to e$ denote the multiplication, and let $m^\vee: e \to e^2$ denote its dual map. The composition $m^\vee \circ \eta $ exhibits $e$ to be self dual, and we define the self duality map to be $s: e \simeq e^\vee$.
We depict the condensation monad of $e$ on $\cA$ by the diagram in \cref{fig:etoA}. 
\begin{figure}
  \centering
 \begin{tikzpicture}[scale=.5][line width=0.9pt]
 
  \def\nw{0.55}     
  \def\ny{0.22}     
  \def\cut{3.7}     
  \def\far{4.8}     
  \def\hx{0.85}     
  \def\hy{0.78}     
 
  \def\rightside{
    (\nw,\cut) -- (\nw,2.75)
    .. controls (\nw,2.30) and (1.15,2.47) .. (1.75,2.18)
    .. controls (2.45,1.80) and (2.80,1.10) .. (2.80,0)
    .. controls (2.80,-1.10) and (2.45,-1.80) .. (1.75,-2.18)
    .. controls (1.15,-2.47) and (\nw,-2.30) .. (\nw,-2.75)
    -- (\nw,-\cut)}
  \def\leftside{
    (-\nw,-\cut) -- (-\nw,-2.75)
    .. controls (-\nw,-2.30) and (-1.15,-2.47) .. (-1.75,-2.18)
    .. controls (-2.45,-1.80) and (-2.80,-1.10) .. (-2.80,0)
    .. controls (-2.80,1.10) and (-2.45,1.80) .. (-1.75,2.18)
    .. controls (-1.15,2.47) and (-\nw,2.30) .. (-\nw,2.75)
    -- (-\nw,\cut)}
 
  \begin{scope}[even odd rule]
    \clip \rightside -- \leftside -- cycle (0,0) ellipse ({\hx} and {\hy});
    \node[inner sep=0, outer sep=0] at (0,0) {\pgfuseshading{surfshade}};
  \end{scope}
 
  \begin{scope}[even odd rule]
    \clip \rightside -- \leftside -- cycle (0,0) ellipse ({\hx} and {\hy});
    \begin{scope}[yscale=0.9]
      \foreach \a in {0,180}{
        \draw[line width=0.65pt, red!25!black, rotate=\a]
             (3.10,0) arc[start angle=0, end angle=180,
                          x radius=1.15, y radius=0.16];
        \draw[line width=0.65pt, red!25!black, dashed, dash pattern=on 2pt off 2pt,
              rotate=\a]
             (3.10,0) arc[start angle=0, end angle=-180,
                          x radius=1.15, y radius=0.16];
      }
    \end{scope}
  \end{scope}
 
  \draw \rightside;
  \draw \leftside;
  \draw (0,0) ellipse ({\hx} and {\hy});
 
  \foreach \s in {1,-1}{
    \fill[red!20!white] (0,\s*\cut) ellipse ({\nw} and {\ny});
    \draw (0,\s*\cut) ellipse ({\nw} and {\ny});
    \draw[dashed] (-\nw,\s*\cut) -- (-\nw,\s*\far);
    \draw[dashed] ( \nw,\s*\cut) -- ( \nw,\s*\far);
    \draw[dashed] (0,\s*\far) ellipse ({\nw} and {\ny});
     \node at (1.2,-3.9) {${\eta^\vee}$};
     \node at (1.2,3.8) {${\eta}$};
          \node at (0,2) {${m}^\vee$};
            \node at (.2,-2) {${m}$};
  }
 
\end{tikzpicture}
 
 \caption{The 1-morphism $e$ is the two-dimensional surface shaded in red and outlined by the black line. Let the three-dimensional space outside of the surface represent $\cA$, as well as the space inside the surface. It is important to notice that the dimensions of the regions do \textbf{not} reflect the spacetime dimension. In particular, since $\cA$ is a topological 2-net it can only assign algebras to 2-disks. Instead, this figure shows the structure of the 3-category $\Net^\fd_2$, for which the objects (that decorate 3-dimensional regions) are 2-nets. Therefore, $e:\cA \to \cA$ is a 1-morphism from the outside to the inside region.
 The dotted lines at the top and bottom of the figure denote the morphism $\id_\cA$ which should be thought of as the invisible surface of a tube.  The two solid circles at the top and bottom denote  the 2-morphisms $\eta:\id_\cA \to e$ and $\eta^\vee:e \to \id_\cA$, respectively. Lastly, the morphisms $m^\vee$ and $m$ are realized via the pair of pants at the  top and bottom of the diagram respectively. }
    \label{fig:etoA}
\end{figure}
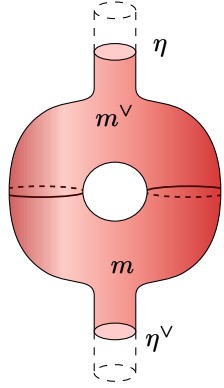

Let $\theta:\id_\cA \to e$ be a 2-morphism which contains all the simple 2-morphisms $\id_\cA \to e$ as direct summands. The dual 2-morphism $\theta^\vee: e \to \id_\cA$ induces a 2-morphism  $\theta^*:\id_\cA \to e$ via the following relations. The evaluation/coevaluation gives an equivalence of 2-morphism spaces
\begin{equation*}
    \hom(e,\id_\cA) \simeq \hom(\id_\cA,e^\vee).
\end{equation*}
We take $\theta^\vee \in \hom(e,\id_\cA) $, from which we can define another 2-morphism via the composition 
\begin{equation}
\overline{\theta^\vee}: (\id_{e^\vee}\circ \theta^\vee)\circ \mathrm{coev}_e: \id_\cA \to e^\vee\,.
\end{equation}
Using the self duality $s_*: \hom(\id_\cA,e^\vee) \to \hom(\id_\cA,e)$, we define  $\theta^* = s_* \circ \overline{\theta^\vee}$.


Consider a finite set of 1-strata, i.e., 1-dimensional slices $R = \{r_1,r_2,\ldots,r_k\}\subset (0,1)^2$, that are in cyclic order and  wrap the ``$a$-cycle'' of the  torus as in \cref{fig:insertedstrata}. 
\begin{figure}
    \centering
  \begin{tikzpicture}[scale=.5][line width=0.9pt]
 
  \def\hx{0.85}     
  \def\hy{0.78}     
 
  \def\body{
    (2.80,0)
    .. controls (2.80,1.10) and (2.45,1.80) .. (1.75,2.18)
    .. controls (1.15,2.51) and (0.60,2.60) .. (0,2.60)
    .. controls (-0.60,2.60) and (-1.15,2.51) .. (-1.75,2.18)
    .. controls (-2.45,1.80) and (-2.80,1.10) .. (-2.80,0)
    .. controls (-2.80,-1.10) and (-2.45,-1.80) .. (-1.75,-2.18)
    .. controls (-1.15,-2.51) and (-0.60,-2.60) .. (0,-2.60)
    .. controls (0.60,-2.60) and (1.15,-2.51) .. (1.75,-2.18)
    .. controls (2.45,-1.80) and (2.80,-1.10) .. (2.80,0)}
 
  \path[use as bounding box] \body -- cycle;
 
  \begin{scope}[even odd rule]
    \clip \body -- cycle (0,0) ellipse ({\hx} and {\hy});
    \node[inner sep=0, outer sep=0] at (0,0) {\pgfuseshading{surfshade}};
  \end{scope}
 
  \begin{scope}[even odd rule]
    \clip \body -- cycle (0,0) ellipse ({\hx} and {\hy});
    \begin{scope}[yscale=0.9]
      \foreach \a in {0,60,120,180,240,300}{
        \meridian{\a-15}{blue!80!black}
        \meridian{\a+15}{blue!80!black}
        \meridian{\a}{black, line width=1pt}
      }
    \end{scope}
  \end{scope}
 
  \draw \body -- cycle;
  \draw (0,0) ellipse ({\hx} and {\hy});
  \node[font=\tiny,scale = .8] at (1.5,2.6) {${r_1}$};
   \node[font=\tiny,scale = .8] at (-1.5,2.6) {${r_2}$};
   \node[font=\tiny,scale = .8] at (-3.1,0) {${r_3}$};
   \node[font=\tiny,scale = .8] at (-1.5,-2.6) {${r_4}$};
   \node[font=\tiny,scale = .8] at (1.5,-2.6) {${r_5}$};
   \node[font=\tiny,scale = .8] at (3.1,0) {${r_6}$};
   \node[font=\tiny,scale = .8] at (2.3,2.1) {${\theta}$};
   \node[font=\tiny,scale = .8] at (.9,2.8) {${\theta^*}$};
      \node[font=\tiny,scale = .8] at (-.9,2.75) {${\theta}$};
       \node[font=\tiny,scale = .8] at (-2.3,2.1) {${\theta^*}$};
       \node[font=\tiny,scale = .8] at (-3,.8) {${\theta}$};
       \node[font=\tiny,scale = .8] at (-3,-.8) {${\theta^*}$};
       \node[font=\tiny,scale = .8] at (-2.3,-2.1) {${\theta}$};
             \node[font=\tiny,scale = .8] at (-.9,-2.75) {${\theta^*}$};
              \node[font=\tiny,scale = .8] at (.9,-2.8) {${\theta}$};
              \node[font=\tiny,scale = .8] at (2.3,-2.1) {${\theta^*}$};
               \node[font=\tiny,scale = .8] at (3,-.8) {${\theta}$};
                \node[font=\tiny,scale = .8] at (3,.8) {${\theta^*}$};
\end{tikzpicture}
\hspace{15mm}
\begin{tikzpicture}[scale = .5][line width=0.9pt]
 
  \def\hx{0.85}     
  \def\hy{0.78}     
 
  \def\body{
    (2.80,0)
    .. controls (2.80,1.10) and (2.45,1.80) .. (1.75,2.18)
    .. controls (1.15,2.51) and (0.60,2.60) .. (0,2.60)
    .. controls (-0.60,2.60) and (-1.15,2.51) .. (-1.75,2.18)
    .. controls (-2.45,1.80) and (-2.80,1.10) .. (-2.80,0)
    .. controls (-2.80,-1.10) and (-2.45,-1.80) .. (-1.75,-2.18)
    .. controls (-1.15,-2.51) and (-0.60,-2.60) .. (0,-2.60)
    .. controls (0.60,-2.60) and (1.15,-2.51) .. (1.75,-2.18)
    .. controls (2.45,-1.80) and (2.80,-1.10) .. (2.80,0)}
 
  \begin{pgfinterruptboundingbox}
 
  \begin{scope}[even odd rule]
    \clip \body -- cycle (0,0) ellipse ({\hx} and {\hy});
    \node[inner sep=0, outer sep=0] at (0,0) {\pgfuseshading{surfshade}};
  \end{scope}
 
  \begin{scope}[even odd rule]
    \clip \body -- cycle (0,0) ellipse ({\hx} and {\hy});
    \begin{scope}[yscale=0.9]
      \foreach \a in {0,60,120,180,240,300}{
        \meridian{\a-15}{blue!80!black}
        \meridian{\a+15}{blue!80!black}
        \meridian{\a}{black, line width=1pt}
      }

    \end{scope}
  \end{scope}
 
  \end{pgfinterruptboundingbox}
 
  \draw \body -- cycle;
  \draw (0,0) ellipse ({\hx} and {\hy});
  \node[font=\tiny,scale = .8] at (1.5,2.6) {${r_1}$};
   \node[font=\tiny,scale = .8] at (-1.5,2.6) {${r_2}$};
   \node[font=\tiny,scale = .8] at (-3.1,0) {${r_3}$};
   \node[font=\tiny,scale = .8] at (-1.5,-2.6) {${r_4}$};
   \node[font=\tiny,scale = .8] at (1.5,-2.6) {${r_5}$};
   \node[font=\tiny,scale = .8] at (3.1,0) {${r_6}$};
   \node[font=\tiny,scale = .8] at (2.3,2.1) {${\theta}$};
   \node[font=\tiny,scale = .8] at (.9,2.8) {${\theta^*}$};
      \node[font=\tiny,scale = .8] at (-.9,2.75) {${\theta}$};
       \node[font=\tiny,scale = .8] at (-2.3,2.1) {${\theta^*}$};
       \node[font=\tiny,scale = .8] at (-3,.8) {${\theta}$};
       \node[font=\tiny,scale = .8] at (-3,-.8) {${\theta^*}$};
       \node[font=\tiny,scale = .8] at (-2.3,-2.1) {${\theta}$};
             \node[font=\tiny,scale = .8] at (-.9,-2.75) {${\theta^*}$};
              \node[font=\tiny,scale = .8] at (.9,-2.8) {${\theta}$};
              \node[font=\tiny,scale = .8] at (2.3,-2.1) {${\theta^*}$};
               \node[font=\tiny,scale = .8] at (3,-.8) {${\theta}$};
                \node[font=\tiny,scale = .8] at (3,.8) {${\theta^*}$};
  \draw ({3.40*cos(45)},{3.20*sin(45)})
        arc[start angle=45, end angle=200, x radius=3.40, y radius=3.20];
  \node at ({3.85*cos(122.5)},{3.62*sin(122.5)}) {$I$};
 
\end{tikzpicture}
    \caption{Left: We have displayed the set $R$ with $k=6$ elements, which are shown by the solid black lines that wrap the handle of the surface of the torus. The blue lines depict insertions of $\theta$ and $\theta^*$ for each $r_i$.  Right: The 2-disk region $I$ can be thought of as the disk region in $e$ spanned by the indicated arc.}
    \label{fig:insertedstrata}
\end{figure}
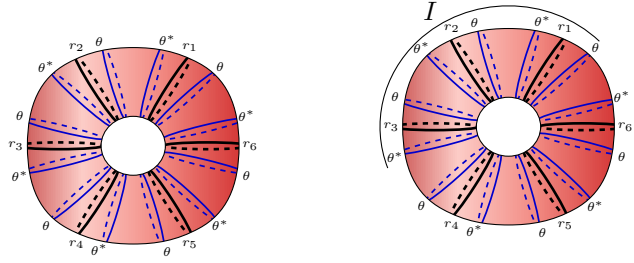
In the neighborhood of each strata we insert $\theta$ and $\theta^*$ and define a sector with respect to $R$ via the following composition:
\begin{equation}
\begin{tikzcd}
     H_R:\id_\cA \arrow[r,"(\theta \theta^*)^k"]&  e^{2k} \arrow[r,"m^k"]&  e \arrow[r, "\eta^\vee"]& \id_\cA\,,
\end{tikzcd}
\end{equation}
where $e^{2k}$ is $2k$ compositions of $e$. For any inclusion of refinements $R \subset S$ that cover $e$ with more 1-strata, the embedding $\eta \hookrightarrow \theta$ induces an embedding $H_R \hookrightarrow H_S$. We take the colimit in the category of sectors of $\cA$ to obtain
\begin{equation}\label{eq:limit}
    H_\cB := \varinjlim_{R} H_R\,.\footnote{Intuitively, one should view this colimit as effectively covering the surface by denser and denser 1-strata of $\theta$ and $\theta^*$. The fact that this produces a condensate for $(\cA,e,\ldots)$ is analogous  to the usual intuition of condensation which builds a phase from a fine mesh of lower dimensional operators \cite[Section 2.4]{Gaiotto:2019xmp}. This a priori produces a sector belonging in $\Net_n$ rather than $\Net^{\fd}_n$. But because quantum spin systems should be defined locally, finiteness is a property that is checked for any bounded disk region $I$ and finite set of 1-dimensional strata $R$.}
\end{equation}
In contrast to the role played by the surface of the torus in \cref{fig:etoA}, the torus in \cref{fig:insertedstrata} now serves as a 2-dimensional space on which we will construct the condensate $\cB$ of $\cA$ on.\footnote{This shift in perspective is due to shifting from thinking about the 2-nets as objects of a 3-category, to actually specifying how the 2-net is defined i.e. how it assigns algebras to 2-disk regions. It follows consistently that in defining the 2-net, its respective sector should live on codimension-1 slices of the surface, and we see this is the case from how $H_R$ is defined via $\theta$ and $\theta^*$.}
We now address how $H_\cB$ is  evaluated specifically on 2-disk regions. This means that for $I\in \Disk^2$ on the surface labeled by $e$ on the right side of \cref{fig:insertedstrata}, we will construct $H_{R,I}$ and from this sector we will produce the algebra $\cB(I)$. The regions $I$ and $I'$ divide the insertions of $\theta$ and $\theta^*$ into two disjoint sets. This also induces a decomposition of  $H_R$ by two sectors labeling the following maps
\begin{equation}
\begin{tikzcd}
     H_R:\id_\cA \arrow[r,"H_{R,I}"]&  e \arrow[r,"\overline{H}_{R,I'}"]& \id_\cA\,,
\end{tikzcd}
\end{equation}
where $H_{R,I}$ is formed from composing $(\theta \theta^*)^j$ with $m^j$, and $j<k$ is the number of strata that is included in $I$.
In taking the colimit as in \cref{eq:limit} we obtain a decomposition
\begin{equation}
    \begin{tikzcd}
     H_\cB:\id_\cA \arrow[r,"H_{\cB,I}"]&  e \arrow[r,"\overline{H}_{\cB,I'}"]& \id_\cA\,.
\end{tikzcd}
\end{equation}
Given this decomposition of $H_\cB$ by the chosen interval $I$, we define the algebra 
\begin{equation}
    \cB(I) := \hom_\cA(H_{\cB,I}, H_{\cB,I})\,.\footnote{Here the hom is taken with respect to the algebra that $\cA$  assigns to 2-disk regions such as $I$.}
\end{equation}
To show that $H_\cB$ satisfies the vacuum property we note that on any region $I$ and finite set $R$ labeling strata, we can form an algebra $\widetilde{\cB}(I) = \hom_\cA(H_{R,I},H_{R,I})$  as well as the sector $H_{\widetilde{\cB}(I)} =  H_{R,I} \boxtimes_{e} \overline{H}_{R,I}$, with the property that $H_{\widetilde{\cB}(I)}$ is the regular $\widetilde{\cB}(I)\text{-}\widetilde{\cB}(I)^{\op}$-bimodule. Taking the colimit over $R$ recovers $\cB(I)$, with its regular $\cB(I)\text{-}\cB(I)^{\op}$-bimodule $H_\cB$.
To conclude this proof, we need to show that $(\cB,H_\cB)$ which is an object of $\Net_2$ has the properties of  a finite semisimple topological $2$-net. For each interval $I$ and  for each finite set $R$, the sector $H_{R,I}$ will be semisimple by the property that $\theta$, $\theta^*$, and $e$ are captured by the objects and morphisms of $\QS_1$ which are semisimple. The duality condition on $H_\cB$ is similarly  checked on finite configurations of strata, where they each give dualizable vector spaces.
\end{proof}

\subsection{Condensation for $n$-nets}\label{subsection:nnet}
Following the inductive approach underlying condensation for higher categories, and building on the construction developed in the preceding subsection as a base case, we assume that $\QS_{n-1}$ is condensation complete and prove that $\QS_n$ is likewise condensation complete.

\begin{thm}\label{thm:QS}
    The category of finite semisimple topological $n$-nets is a symmetric monoidal condensation complete $(n+1)$-category  $\QS_n$, which is equivalent to $\npVect^{\fd}_\mathbb{C}$.
\end{thm}

\begin{proof}
Following the construction in \S\ref{subsection:2net}, we let $\cA$ be a topological $n$-net which now occupies a $(n+1)$-dimensional region, and higher morphisms in $\QS_n$ label strata in higher codimension. In an analogous way, given  a condensation monad $(\cA,e,\ldots)$, the condensation algebra $e$ is defined on a codimension-1 hypersurface.

We introduce $k$-morphisms for $2\leq k <n$ given by $\theta_k: \id_{{\ldots}_{\id_\cA}} \to f_k$ that includes all the simple $k$-morphisms as summands, and assign it to all the faces of codimension-$k$ in the ambient $(n+1)$-dimensional setting. The source for this map is given by taking the identity of the identity of $\cA$ $(k-1)$-times, and the target $f_k$ is a $(k-1)$-morphism obtained by composing the morphisms decorating the faces of lower codimension (up to codimension 1) around a face of codimension-$k$. In particular, if $k=2$ then $\theta_2: \id_\cA \to f_2 = f^\vee \circ f = e$ is what was called $\theta$ in the proof of \cref{prop:QS2}. We require that each $\theta_k$ extends to a unital condensation with $\eta_{\theta_k}:\id_{{\ldots}_{\id_\cA}} \to f_k$, and moreover we introduce insertions of $\theta_k$ on faces of codimension-$k$. The particular $n$-morphism $\theta_n$ induces a map $\theta^*_n: \id_{{\ldots}_{\id_\cA}} \to f_n$, again using the self-duality of $e$ at the top level.

Just like how we inserted lines on the surface that implements $e$ in the proof of \cref{prop:QS2}, we now insert a mesh of strata $\cR$ along the hypersurface for $e$, realizing it as a simplicial complex where each cell has a canonical orientation determined by the ordering of its boundary. Let $\cR^\vee$ denote the dual decomposition of $\cR$, so that the $k$-morphism $\theta_k$ will be assigned to each $(n+1-k)$-simplex of $\cR^\vee$  for $1\leq k \leq n$. Composing all the $k$-morphisms around every face results in a sector $H_\cR: \id_{{\ldots}_{\id_\cA}} \to \id_{{\ldots}_{\id_\cA}}$ of $\cA$ which is unique up to isotopy of $\cR$. Given a finer simplicial complex $\cS$ that realizes the hypersurface, and where $\cS$ is a subdivision of $\cR$, we have an embedding $H_\cR \hookrightarrow H_\cS$ via the inclusions $\eta_{\theta_{k}} \hookrightarrow \theta_k$. Taking a colimit over $\cR$ in the category of sectors for $\cA$ we define
\begin{equation}
 H_\cB := \varinjlim_{\cR} H_\cR\,.
\end{equation}

For an $n$-disk region $I$ on the $n$-dimensional hypersurface, which inherits a simplicial structure induced by $\cR$. We consider $H_{\cR,I}$ to be the composition of those $\theta_k$ inside the region, $\overline{H}_{\cR,I'}$ to be the composition of those $\theta_k$ in  $I'$ -- the complement of $I$, and denote $\Gamma_{\cR,I}$ to be the composition of those $k$-morphisms on the boundary of $I$. Then we have a decomposition 
\begin{equation}
\begin{tikzcd}
      H_\cR: \id_{{\ldots}_{\id_\cA}}   \arrow[r,"H_{\cR,I}"]&  \Gamma_{\cR,I} \arrow[r,"\overline{H}_{\cR,I'}"]& \id_{{\ldots}_{\id_\cA}}\,,
\end{tikzcd}
\end{equation}
which in the colimit over subdivisions $\cR$ yields 
\begin{equation}
\begin{tikzcd}
      H_\cB: \id_{{\ldots}_{\id_\cA}}   \arrow[r,"H_{\cB,I}"]&  \Gamma_{I} \arrow[r,"\overline{H}_{\cB,I'}"]& \id_{{\ldots}_{\id_\cA}}\,.
\end{tikzcd}
\end{equation}
The algebra $\cB$ assigns to $I$ is constructed as $\cB(I) = \hom_\cA(H_{\cB,I},H_{\cB,I})$. One then applies the same arguments in the proof of \cref{prop:QS2} to show that $(\cB,H_\cB)$ exhibits the properties of a finite semisimple topological $n$-net.
\end{proof}

\subsection{The Group of Higher Dimensional QCAs}\label{subsection:higherQCAgroups}

We apply \cref{thm:QS} specifically to $n=4$, i.e. for the category of quantum spin systems in (3+1)d, to identity the  categories $\QS_4 \simeq \mathbf{5Vect}^{\fd}_\C$. Consider the Picard groupoid $(\mathbf{5Vect}^{\fd}_\C)^\times$, whose
 homotopy groups are computed in \cite{johnson2022classification,johnson2024minimal,Bhardwaj:2024xcx}. In particular, thanks to the delooping procedure used to construct $\QS_4$, the group of invertible 1-morphisms is given by $\pi_1((\mathbf{5Vect}^{\fd}_\C)^\times) = \pi_0((\mathbf{4Vect}^\fd_\C)^\times)$. 
 The category $\mathbf{4Vect}^\fd_\C$ is equivalent to the Morita category of (multi)fusion 2-categories constructed in \cite{Decoppet:2024htz}, or equivalently the Morita category of braided multifusion 1-categories \cite{brochier2021dualizability,decoppet2022drinfeld}. The invertible braided multifusion 1-categories are $\cA$ such that there exists a category $\cD$ such that $\cA \boxtimes \cD$ is Morita equivalent to a Drinfeld center $\cZ(\cC)$ for some fusion category $\cC$. In this definition, the monoidal structure is defined via the Deligne tensor product of categories.
 The invertible braided multifusion 1-categories are \textit{nondegenerate} braided fusion 1-categories \cite{brochier2021invertible}. 
 
 In the theory of braided fusion categories, there is an appropriate notion of a center  of a braided fusion category $\cB$  known as the Müger center, which keeps track of how far a category is from being nondegenerate. The Müger center  has objects given by 
\begin{equation}
    \cZ_2(\cB):=\{x \in \cB |\, \beta_{y,x}\circ \beta_{x,y} = \mathrm{id}_{x\otimes y},\quad  \forall y \in \cB \}
\end{equation}
where $\beta_{x,y}: x\otimes y \to y \otimes x$ is the braiding on $\cB$. A braided category $\cB$ is nondegenerate if $\cZ_2(\cB) = \mathbf{Vect}$. The group of  Morita classes of nondegenerate braided fusion 1-categories, captured by $\pi_0$, is the Witt group $\Witt$ of nondegenerate braided fusion 1-categories \cite{davydov2010witt} generated by Witt classes of nondegenerate braided fusion 1-categories. Two nondegenerate  braided fusion categories $\cA$ and $\cB$ are said to be in the same Witt class (aka Witt equivalent) if there exists a fusion categories $\cC_1$ and $\cC_2$ such that $\cA \boxtimes \cZ(\cC_1) \simeq \cB \boxtimes \cZ(\cC_2)$ as braided tensor categories. In this definition, the Witt equivalence relation is generated by equating categories which are braided tensor equivalent, and setting Drinfeld centers to be trivial.
With respect to the discussion in \S\ref{subsubsection:spacespectrum}, we should view nondegenerate braided fusion 1-categories as categorified Azumaya algebras, and $\Witt$ as the analog of a higher Brauer group for these algebras. A final upshot of the homotopical nature of the net-based formalism is that $\QS_4$ allows us to identify QCAs in (1+1)d and (2+1)d. By looping further we obtain an analogous category for (2+1)d quantum spin systems and their morphisms. The category of main importance here is $\mathbf{3Vect}^{\fd}_\C$, whose objects are given by separable finite tensor categories over $\C$, i.e. the Morita category of multifusion 1-categories \cite{douglas2021dualizable}.

In light of \cref{notation:X} for defining quantum spin systems 
we now broaden the definition of topological $n$-nets by mapping the ambient cube $(0,1)^n$ used for their definition to $\mathbb R^n$. Furthermore we take the assignment of local algebras to be on the sites of the standard lattice $\mathbb Z^n\subset\mathbb R^n$.
By tracing through the definitions of the higher morphisms in $\QS_4$ we obtain the following descriptions for the invertible morphisms. 

\begin{cor}\label{cor:QCAclassification}
  Let $1 \leq n\leq 3 $. The groups of QCAs  over $\mathbb{C}$ up to quantum circuits and stabilization on the spaces $\Z^{n+1}$ are classified by the following groups:
    \begin{itemize}
        \item For $n=3$, $\cQ(\Z^4)/\cC(\Z^4)= \Witt$: Morita classes of braided fusion 1-categories with trivial Müger center, i.e. the Witt group of nondegenerate braided fusion 1-categories \cite{davydov2010witt}. 
        \item For $n=2$, $\cQ(\Z^3)/\cC(\Z^3)= 0$: Morita classes of commutative algebras over $\C$ with trivial center.\footnote{This group can equivalently be described as braided fusion 0-categories.}
        \item For $n=1$, $\cQ(\Z^2)/\cC(\Z^2)= 0$: Morita classes of central simple algebras over $\C$.
    \end{itemize}
\end{cor}

\begin{rem}\label{rem:fermionicQCA}
    The category $\mathbf{4Vect}^{\fd}_\C$, and its fermionic counterpart $\mathbf{4SVect}^{\fd}_\C$, has played a big role in the study of (3+1)d TQFTs \cite{Johnson-Freyd:2020twl,Decoppet:2025eic,Debray:2025kfg,Kobayashi:2026lgd}. In this sense, topological $n$-nets can also be seen as a necessary framework to provide a  bridge between the lattice and continuum settings while elucidating the homotopical structures underlying both. This perspective also offers guidance for understanding the structure of QCAs in fermionic quantum spin systems. 
    
    One can apply analogous constructions as in this paper to develop a category of fermionic topological $n$-nets $\mathbf{FQS}_n$, which would then describe fermionic quantum spin systems and their QCAs. By analogy, the category $\mathbf{FQS}_4 \simeq \mathbf{5SVect}^\fd_\C$, and thus the homotopy groups $\pi_{i}(\mathbf{4SVect}^{\fd}_\C)^\times$ encode the groups of fermionic QCAs in dimensions $(3-i+1)$d.
    These homotopy groups are given by:
\begin{itemize}
    \item  $\pi_{0}((\mathbf{4SVect}^{\fd}_\C)^\times) = \sWitt$: the Witt group of slightly degenerate braided fusion categories \cite{davydov2013structure}.
    \item $\pi_{1}((\mathbf{4SVect}^{\fd}_\C)^\times) = 0$.
    \item  $\pi_{2}((\mathbf{4SVect}^{\fd}_\C)^\times) = \Z/2$: Morita classes of central simple super algebras over $\C$. The nontrivial class is given by an odd Clifford algebra, Morita equivalent to $\mathrm{Cliff}(1)$. 
\end{itemize}

\end{rem}

\bibliographystyle{alpha}
\bibliography{references}
\end{document}